\documentclass[moor]{informs3_arxiv}              
\OneAndAHalfSpacedXI

\usepackage{natbib}
 \NatBibNumeric
 \def\bibfont{\small}%
 \bibpunct[, ]{[}{]}{,}{n}{}{,}%

 \usepackage{amsmath,amsfonts, mathdots, amssymb, latexsym, amsbsy,bm,mathtools}         
 \usepackage{braket}
 \usepackage[english]{babel}
 
 \graphicspath{{figures/}}

\usepackage[colorlinks=true,breaklinks=true,bookmarks=true,urlcolor=blue,
     citecolor=blue,linkcolor=blue,bookmarksopen=false,draft=false]{hyperref}

\usepackage{cleveref}
\crefname{lemma}{Lemma}{Lemmas}
\crefname{theorem}{Theorem}{Theorems}

\usepackage{thmtools}
\usepackage{thm-restate}

\usepackage{dsfont}

\newcommand{\R}{\mathbb{R}}

\newcommand*\diff{\mathop{}\!\mathrm{d}}

\DeclarePairedDelimiterX{\scalar}[2]{\langle}{\rangle}{#1, #2}

\newcommand{\A}{\mathcal{A}} 

\newcommand{\abs}[1]{\left\lvert #1 \right\rvert}
\newcommand{\NP}{\ensuremath{\mathsf{NP}}}

\newcommand{\new}[1]{#1}

\DeclareMathOperator{\VI}{VI}

\renewcommand{\l}{\ell}

\newcommand{\Time}{[0, \infty)}
\newcommand{\Flow}{\R_{\geq 0}}

\newcommand{\capa}{\nu}

\renewcommand{\epsilon}{\varepsilon}

\newcommand{\Sources}{S}

\newcommand{\source}{s}
\newcommand{\sink}{t}
\newcommand{\Paths}{\mathcal{P}}

\newcommand{\m}{m} 
\newcommand{\J}{J} 
\newcommand{\K}{K} 
\newcommand{\X}{X} 
\newcommand{\exG}{\bar G}
\newcommand{\exE}{\bar E}
\newcommand{\exV}{\bar V}
\newcommand{\exf}{\bar f}

\newcommand{\inrate}{r} 

\newcommand{\T}{T} 

\newcommand{\defemph}[1]{\emph{#1}}

\usepackage{enumitem}
\renewcommand\labelenumi{(\roman{enumi})}
\renewcommand\theenumi\labelenumi

\def\EMAIL#1{\href{mailto:#1}{#1}}
\def\URL#1{\href{#1}{#1}}         

\TheoremsNumberedThrough     

\EquationsNumberedThrough    

\begin{document}


\RUNAUTHOR{Leon Sering}

\RUNTITLE{Multi-Commodity Nash Flows Over Time}

\TITLE{Multi-Commodity Nash Flows Over Time}

\ARTICLEAUTHORS{%
\AUTHOR{Leon Sering \footnote{Funded by the Deutsche Forschungsgemeinschaft (DFG, German Research Foundation) under Germany's Excellence Strategy – The Berlin Mathematics Research Center MATH+ (EXC-2046/1, project ID: 390685689).}}
\AFF{Institute of Mathematics, Technische Universität Berlin, Germany, \EMAIL{sering@math.tu-berlin.de}, \URL{https://www.coga.tu-berlin.de/people/sering}}
} 

\ABSTRACT{%
Motivated by the dynamic traffic assignment problem, we consider flows over time model with deterministic queuing.
Dynamic equilibria, called Nash flows over time, have been studied intensively since their introduction by Koch and
Skutella in 2009. Unfortunately, the original model can only handle a single commodity, i.e., all flow goes from one
origin to one destination, therefore, this theory is not applicable to every-day traffic scenarios. \new{For this reason we
consider Nash flows over time in a setting with multiple commodities, each with a individual origin and destination. We
prove their existence with the help of multi-commodity thin flows with resetting. This proof gives some interesting structural
insights into the strategy profiles underlying the equilibria, as they emphasis the inter-commodity dependencies.
Finally, we show that the two special cases of a common origin or a common destination can be reduced to the
single-commodity case.} }%


\MSCCLASS{Primary: 05C21, 91A13; secondary: 49J40, 90B20} 
\ORMSCLASS{Primary: Networks/graphs: Multicommodity, Games/group decisions:Non-atomic; secondary: Transportation: Road}

\maketitle

%


%
%
%

\section{Introduction} Understanding traffic is a challenging task as traffic dynamics are hard to predict and changes
in the network can have surprising  unintentional effects. For example the congestion in a network can increase when a
new road is added or a bottle-neck road is enlarged even though the demand and total traffic volume has not changed.
This famous effect is known as \emph{Braess Paradox}
\cite{braess1968paradoxon,braess2005paradox,roughgarden2001designing} and it was observed in New York (1990)
\cite{kolata1990if} and Seoul (2005) \cite{vidal2016seoul}. (In fact, the reverse effect was observed in these
instances: A closure of a road improved the traffic situation.) This emphases how important it is to understand the
complicated interplay of traffic users and to predict it as best as possible before applying costly and sometimes
irreversible changes to the infrastructure.

In order to understand the traffic dynamics mathematically, we need to formulate a strong traffic model, in which, on
the one hand, it is possible to provide provable statements and predictions and which, on the other hand, represents
real-world traffic in the best possible way. Of course, different aspects of traffic have been considered from the
mathematical perspective for a long time and there are several different approaches for modeling traffic
\cite{correa2004computational,dafermos1969traffic,may1990traffic,newell1955mathematical,wardrop1952road}. But most of
them are either too simplified, leading to large discrepancy to real-world traffic, or they are too complicated to prove
any results for larger networks. Only in recent years there was a significant scientific improvement by combining the
main concepts of previous static models with a continuous time approach \cite{koch2011nash}.

In almost every case the traffic is modeled as a network flow on a graph and the traffic assignment problem is to find a
traffic flow, for which every part that corresponds to a traffic user is traveling on a reasonable path from its origin
to its destination. In simpler models time is not taken into account but instead it is assumed that the flow represents
the average traffic on a constant demand. On the other side of the spectrum, we have models that aim to simulate the
dynamics on each road segment in great detail including acceleration, deceleration, reaction times and distances between
vehicles. In this article we consider a flow over time model that sits in the middle. It has a continuous time component
in order to represent the dynamic flow evolution, as some of the most interesting phenomena can only be observed in such
dynamic models. But at the same time we stay on a macroscopic level and assume a straightforward queuing model on each
link, which does not take detailed vehicle dynamics into account.

\begin{figure}[b]
\centering \includegraphics{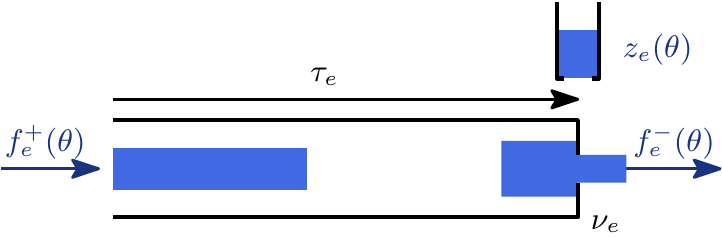} 
\caption{The deterministic queuing model. Every arc $e$ is equipped with a transit time $\tau_e$ and a capacity $\capa_e$. The total in- and outflow flow rate of flow entering and leaving arc $e$ are denoted by $f_e^+(\theta)$ and $f_e^-(\theta)$ for every point
in time $\theta \in \Time$. Additionally, we consider the queue size $z_e(\theta)$ at time $\theta$, as only a rate of
$\capa_e$ can leave the arc at any point in time and excessive flow has to wait in the queue.} \label{fig:road:base}
\end{figure}%

In the \emph{deterministic queuing model} the traffic is described by a continuous flow over time in a directed graph. Each
street segment corresponds to a directed arc between two nodes; see \Cref{fig:road:base}. Every arc $e$ is hereby
equipped with a free flow transit time $\tau_e$ and a capacity $\capa_e$. Every infinitesimal small flow particle first have to traverse the
arc, which takes $\tau_e$ time. If the flow rate that wants to leave the arc exceeds the rate capacity $\capa_e$ a point
queue is built up right before the head of the arc, where the excessive particles line up. Within the queue the particle follow
the first-in-first-out (FIFO) principle, which means that no particle can overtake any other particle and as long as
there is a positive queue the outflow rate operates at capacity rate.

\new{A dynamic equilibrium in the deterministic queuing model is a flow over time where almost all particles travel
along a fastest route to its destination. Since these so-called \emph{Nash flows over time} were introduced by Koch and
Skutella \cite{koch2011nash}, a quite active research line has been emerged (see related work section). But there has
always been one tremendous drawback of this model: Almost all of the structural insights and constructive results only
hold for a single commodity, which means that all flow particles, corresponding to the traffic users, share the same
origin and the same destination. There was some initial approaches to transfer the results to a multi-terminal setting
\cite{sering2018multiterminal}, where flow particles only share either the origin or the destination. But still, all the traffic
must basically go in the same direction. Clearly, this is highly unrealistic when considering daily real-world traffic
and there are only very few setting where one commodity can be justified (some application, for example, are big events
outside of town where everyone is going to, or evacuation scenarios where every safe spot can be connected to a
artificial super sink). This article is dedicated to extend the structure of Nash flows over time to a multi-commodity
setting, where each commodity has its individual origin/source and individual destination/sink.}

\new{So far this setting was only considered by Cominetti et al. in \cite{cominetti2015dynamic}, where they show the
existence of dynamic equilibria in a multi-commodity setting. But unfortunately a lot of details are missing in this
paper (e.g., it not clearly stated how the outflow rate depend on the inflow rates for each commodity) and since the
proof is done via a path-based formulation it does not give much insights into the structure of dynamic equilibria.
Nonetheless, this paper can be seen as a continuation of their work and the important part of this paper is based on the
techniques used within their existence proof.}
 
\begin{figure}[b]
\centering \includegraphics[page=1]{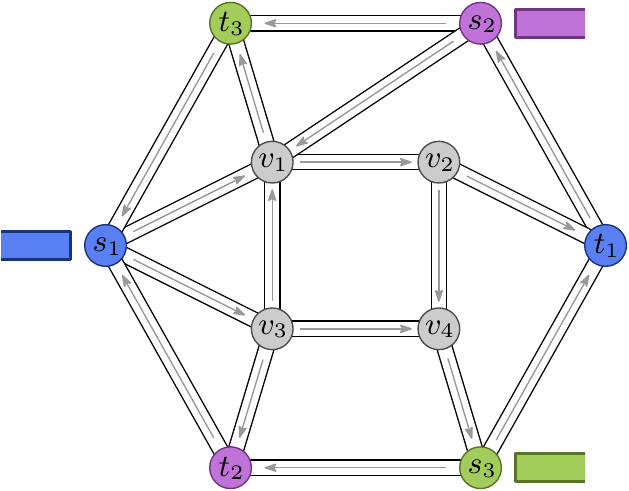} \qquad \quad
\includegraphics[page=2]{setting_multi_commodity} \caption{\emph{On the left:} A multi-commodity network with three
commodities. \emph{On the right:} The particles of each commodity can choose from at least two paths. Each path overlaps
with possible paths of other commodities. Hence, the waiting times a particle experiences on these links do not only
depend on the flow in front of the same commodity but also on the route choices of the flows of other commodities, which
may even enter the network at a later time than the particle itself. For example, the waiting time that particle $\phi =
0$ of commodity $1$ experiences on arc $\source_3\sink_1$, depends on the flow of commodity~$3$ entering at a later
point in time. Their decisions, however, depend on the congestion on arc $\source_1v_1$, which in turn might be
congested by flow of commodity $1$ that has entered later than particle $\phi$.}
\label{fig:setting:multi_commoidty}
\end{figure}
 
\paragraph{Multiple origin-destination-pairs.} In order to give a bit more details on the challenges of this problem, let us
consider a flow over time consisting of multiple commodities $\J$, each of them with its own
origin-destination-pair~$(\source_j, \sink_j)$ and its own network inflow rate $\inrate_j \geq 0$; see
\Cref{fig:setting:multi_commoidty}. At every origin~$\source_j$, a flow enters the network with rate~$\inrate_j$ during
some interval $I_j$ and every infinitesimally small particle of this flow has the goal to reach destination~$\sink_j$ as
early as possible, while considering the queuing delays on the paths. Every commodity is modeled by time-dependent in-
and outflow rates for every arc that must satisfy flow conservation at every node individually. Then, a dynamic
equilibrium consists of a multi-commodity flow over time, where each particle chooses a convex combination of fastest
routes from $\source_j$ to~$\sink_j$ as strategy. Unfortunately, the techniques used for single-commodity settings are
not sufficient for analyzing or algorithmically constructing such multi-commodity Nash flows over time. The fact that
each commodity has different earliest arrival times is the main difficulty as this causes cyclic interdependencies. Each
particle entering the network has to take into account not only all flow that previously entered the network, but also
flow entering the network in the future. This challenging situation is further specified in the example in
\Cref{fig:setting:multi_commoidty}.

\subsection{Related work.}
\paragraph{Dynamic traffic assignment.} Network loading models in time-dependent networks have been studied intensively
within the traffic science community and the several different approaches can be classified into three categories
depending on their level of detail. On the \defemph{macroscopic} level, traffic is modeled by a flow representing a
collection of vehicles. Some of the pioneer work in this regard is due to Vickrey with his single link-load model
\cite{vickrey1969congestion} and to Mechant and Nemhauser with an exit-function-based flow model
\cite{merchant1978model}. Recently, more advanced approaches like the Colombo phase transition model
\cite{blandin2011general,colombo2002} or the macroscopic node model \cite{flotterod2011operational,smits2015family} have
become popular. In contrast to this, \defemph{microscopic} traffic flow models consider each vehicle individually and
track not only the position, speed and acceleration of each car, but they also simulate maneuvers, like lane changes and
overtaking. Some of them even consider different driver behavior, such as gap-acceptance, reaction times and more. For
further details on this kind of models we refer to the comprehensive surveys of Algers et al.\ \cite{algers1997review}
and Olstam and Tapani \cite{olstam2004comparison}. Finally, \defemph{mesoscopic} models are in between these two. Some
aspects, such as the traffic dynamics, are considered at a low level of detail (macroscopic) whereas other aspects, such
as the agent behavior and the route choice for each traffic user, are considered individually (microscopic). This gives a
trade-off between accuracy and computational complexity. Examples of such models are DynaSMART
\cite{jayakrishnan1994evaluation} and the multi-agent transport simulation (MATSim) \cite{horni2016multi}. For a more
detailed overview of different dynamic traffic assignment models we refer to the book of Ran and Boyce \cite{ranboyce96}
and the survey article of Wang et al.\ \cite{wang2018dynamic}. 

\paragraph{Flows over time.} Classical network flows, or \defemph{static flows} as we call them in this article, have
been studied from a mathematical perspective since the middle of the last century. Ford and Fulkerson did a lot of
pioneer work on these structures and they also were the first to introduce \defemph{dynamic flows}, also called
\defemph{flows over time}, back in 1958 \cite{ford1958constructing,ford1962flows}. In a flow over time, every flow
particle travels over time through a network, and therefore, it is an excellent basis for a traffic model. Considering a
network with a single source and a single sink as well as a capacity and a transit time for each arc, Ford and Fulkerson
showed how to efficiently construct a \defemph{maximum flow over time} for some given time horizon. Hereby, a maximum
flow over time is a flow over time that sends as much flow volume as possible from the source to the sink within the
time horizon. Their algorithm is based on a static min-cost flow computation in the given network, where arc transit times
are interpreted as costs. The resulting static flow corresponds to the flow rates of a maximum flow over time
that needs to be sent into the network and along the paths as long as possible.

Closely related to the maximum flow over time problem is the \defemph{quickest flow problem}. Here, a specific amount of flow
volume is given and the task is to send all of it as quickly as possible to the sink. This can be solved efficiently by
using Ford's and Fulkerson's algorithm in combination with a binary search framework \cite{fleischer1998efficient}, which
can be improved to a strongly-polynomial running-time by using parametric search \cite{burkard1993quickest}.

Surprisingly, it is possible, at least for the single-source single-sink setting, to compute a flow over time that is
maximal for all time horizons simultaneously, and that is, therefore, also a quickest flow for all given flow volumes at
once. Such special flows are called \defemph{earliest arrival flows} and their existence was already shown by Gale in
1959 \cite{gale1959transient}. Just as for all other flow over time concepts, earliest arrival flows were first
considered in a discrete time model and, only in 1998, they were extended to a continuous time model by Fleischer and Tardos
\cite{fleischer1998efficient}. Minieka showed in 1973 that they can be computed by using the successive shortest path
algorithm, where it is also allowed to send negative flow backwards in time in the opposite direction of an arc
\cite{minieka1973maximal}.  Unfortunately, this might take an exponential number of iterations in general, which was
shown by Zadeh \cite{zadeh1973bad}, and Disser and Skutella showed that it is $\NP$-hard to compute them
\cite{disser2019simplex}. When considering networks with several sinks, it is possible that there does not exist an
earliest arrival flow anymore~\cite{fleischer2001faster}. In fact, it is $\NP$-hard to decide whether such a flow
exists or not in these networks \cite{schloeter2019earliest}.

\paragraph{Nash flows over time.} All flow over time problems discussed so far are based on the assumption that the
total flow is controlled by a central authority, which decides on the route and departure time of each single particle. In
real-world traffic situations, however, each traffic user acts independently and selfishly, and therefore, we have a
lack of coordination. To capture this behavior, we assume that each flow particle is an individual agent that wants to
reach its destination as early as possible, and hence, we consider flows over time from a game-theoretical perspective.
In this article we study dynamic equilibria, which are states where no particle can reach the destination earlier by
unilaterally changing its route. Hereby, the arc dynamics are described by the \defemph{deterministic queuing model},
which was first mentioned by Vickrey in 1969~\cite{vickrey1969congestion} and studied by Hendrickson and Kocur in
1981~\cite{hendrickson1981schedule}.

In 2009, Koch and Skutella characterized the structure of dynamic equilibria in a single-source single-sink network from
a strictly mathematical point of view~\cite{koch2011nash}; see also Koch's PhD thesis~\cite{koch2012phd}. As the most
essential structural insight, they prove that these equilibrium flows, called \defemph{Nash flows over time}, consist of
a number of phases, in which all flow entering the network chooses the same routes from the source to the sink. Each
phase is, thereby, characterized by the strategy of the particles in form of static flows featuring specific properties,
which they called \defemph{thin flows with resetting}. Based on this key observation, Cominetti, Correa and Larr\'e
showed existence and uniqueness of these thin flows with resetting, and thus, proved the existence of Nash flows over
time~\cite{cominetti2011existence}. They extended this existence result in 2015 to networks with general inflow rate
functions and also to a multi-commodity setting \cite{cominetti2015dynamic}. 
\new{This result is of particular importance for
this paper as we use similar techniques including infinite-dimensional variational inequalities and the
existence theorem of Br\'ezis \cite{brezis1968}. In collaboration with Skutella we showed
how to transfer the constructive approach from single-commodity Nash flows over time to a multi-commodity setting, where
each commodity either share the same origin or the same destination \cite{sering2018multiterminal}. How these approaches
are connected with the results of this paper are shown in \Cref{sec:common_destination_or_origin}.}

Moreover, Macko, Larson and Steskal showed the existence of the Braess Paradox in the single-commodity model~\cite{macko2013braess}, and
Cominetti, Correa and Olver examined the long-term behavior of queues and were able to bound their lengths whenever the
network capacity is sufficiently large~\cite{cominetti2017long}. \new{In a collaboration with Vargas Koch we
extended the model to represent spillback effects by assigning a storage capacity to each arc and
describing the flow dynamics for full arcs \cite{sering2019nash}}

The price of anarchy in this model measures the increase of the arrival times of the particles in a Nash flow over time compared to an earliest
arrival flow. Bhaskar, Fleischer and Anshelevich showed that this ratio can be bounded by $\frac{e}{e - 1}$ under some
very specific conditions on the network~\cite{bhaskar2015stackelberg}. Recently, Correa, Cristi and Oosterwijk
reduced these preconditions significantly \cite{correa2019price}. But the conjecture, that the price of anarchy is
$\frac{e}{e - 1}$ in general, remains open.

\subsection{Our contribution.} The biggest drawback of the Nash flow over time model was that we can only consider a
single-commodity, which corresponds to a dynamic traffic assignment, where all road users either start at the same
origin or everyone want to get to the same destination. To attack this flaw, we consider \emph{multi-commodity flows
over time}, where each commodity has its own origin-destination pair. Even though the existence of these dynamic
equilibria was known before, we show that they can still be characterized by some extended thin flow formulation, which
provides a much more structural, and in particular, edge-based proof for their existence. The key idea is to take  all
flow from the past and the future into consideration at once and to incorporate the flow of other commodities, called
\emph{foreign flow}, into the thin flow definition. It is then possible to prove that multi-commodity Nash flows over
time correspond one-to-one to multi-commodity thin flows. Finally, the existence of these thin flows can be shown by a
reformulation to a infinite-dimensional variational inequality and the existence theorem of Br\'ezis \cite[Theorem 24]{brezis1968}. As an additional
result we show that the multi-terminal Nash flows over time introduced in \cite{sering2018multiterminal} are
multi-commodity Nash flows over time, where all commodities either share the same origin but have different destinations
or have potentially different origins but share the same destination.

%

\subsection{Outline.}
\new{After specifying a proper
multi-commodity flow over time model in \Cref{sec:flow_dynamics:multi_commodity}, we define dynamic equilibria in this setting in \Cref{sec:nash_flows:multi-commodity}. Afterwards, in \Cref{sec:thin_flows:multi_commodity}, we consider
multi-commodity thin flows and prove that the derivatives of multi-commodity Nash flows over time satisfy these flow
conditions. As the main result, we prove the existence of multi-commodity Nash flows over time in \Cref{sec:existence_of_multi_commodity_nash_flows:multi} using the multi-commodity thin flows.
In \Cref{sec:common_destination_or_origin} we recall the multi-terminal Nash flows over time from \cite{sering2018multiterminal} and show that they are indeed satisfy our definition for multi-commodity Nash flows over time. Finally, in \Cref{sec:conclusion}, we give a brief conclusion and some remaining open problem for further research.}


\section{Notation} \label{sec:notation:prelim} In order to make this article as comprehensible as possible we try to
keep the notation intuitive and consistent. Throughout this paper $f$ is used for a flow over time consisting of
functions representing flow rates that are (locally) integrable but in general not differentiable. For the integral
functions, which are almost everywhere differentiable, we use the corresponding capital letter $F$. Graphs are denoted
by $G = (V, E)$ with source $\source$ and sink $\sink$ and the letters $u, v, w$ are reserved for nodes, whereas~$e$ is
used for arcs. The non-negative real numbers are normally denoted by $[0, \infty)$. This notation is used whenever we
denote flow rates or points in time (denoted by $\theta$ or $\vartheta$) and only if we consider the flow of a
commodity, consisting of infinitesimal small particles denoted by $\phi$ or $\varphi$, we write~$\Flow$. The set of
commodities is given by $\J$ and the individual commodities are denoted by either $j$ or $i$. Finally, $x$, $y$ and $\l$
are reserved for the underlying static flow, the underlying foreign flow and the earliest arrival times. Their
derivatives are denoted by $x'$, $y'$ and $\l'$.

\section{Flow dynamics.} \label{sec:flow_dynamics:multi_commodity} For our purposes a network consists of a directed graph $G = (V,
E)$, where each arc~$e$ is equipped with a transit time~$\tau_e \geq 0$ and a capacity~$\capa_e > 0$. In addition,
we have given a finite set of commodities $\J$, each of them equipped with an origin-destination-pair $(s_j, t_j) \in V^2$
and with a network inflow rate $\inrate_j > 0$ as well as a finite time interval $I_j = [a_j, b_j) \subset [0, \infty)$. We assume that there exists at least one
$\source_j$-$\sink_j$-path for every $j \in \J$.

\paragraph{Multi-commodity flows over time.} For a multi-commodity flow over time we consider a family of locally
integrable and bounded functions $f = (f_{j, e}^+, f_{j, e}^-)_{j \in \J, e \in E}$ where $f_{j, e}^+(\theta)$ denotes
the \defemph{inflow rate} of commodity $j$ into arc $e$ at time $\theta$ and $f_{j, e}^-(\theta)$ denotes the respective
\defemph{outflow rate}. The \defemph{cumulative in- and outflow} for each commodity $j$ and each arc $e$ is defined by
\[F^+_{j,e}(\theta) \coloneqq \int_{0}^{\theta} f_{j,e}^+(\xi) \diff \xi \quad \text{ and } 
\quad F^-_{j, e}(\theta) \coloneqq \int_{0}^{\theta} f_{j, e}^-(\xi) \diff \xi,\]
and the \defemph{total (cumulative) in- and outflow rates} at each point in time $\theta$ are given by 
\[f_e^+(\theta) \coloneqq \sum_{j \in \J} f_{j,e}^+(\theta), \quad f_e^-(\theta) 
\coloneqq \sum_{j \in \J} f_{j, e}^-(\theta), \quad F_e^+(\theta) \coloneqq \sum_{j \in \J} F_{j,e}^+(\theta) 
\quad \text{ and } \quad F_e^-(\theta) \coloneqq \sum_{j \in \J} F_{j, e}^-(\theta).\]
 
\paragraph{Flow conservation.} We say $f$ is a \defemph{multi-commodity flow over time} if every commodity $j \in \J$
conserves flow on every arc~$e$:
\begin{equation} \label{eq:conservation_on_arc:multi_commodity}
F_{j,e}^-(\theta + \tau_e) \leq F_{j,e}^+(\theta) \qquad \text{ for all } \theta \in \Time,
\end{equation}
and conserves flow at every node $v \in V \setminus \set{\sink_j}$:
 \begin{equation} \label{eq:flow_conservation:multi_commodity}
 \sum_{e\in \delta^+_v} f_{j, e}^+(\theta) - \sum_{e \in \delta^-_v} f_{j, e}^-(\theta) = \begin{cases}
 0 & \text{ if } v \in V\setminus \set{\source_j} \text{ or } \theta \notin I_j \\
 r_j & \text{ if } v = \source_j \text{ and } \theta \in I_j,
 \end{cases} \quad \text{ for almost all } \theta \in \Time.\end{equation}
 
 Particles entering arc $e$ at time $\theta$ need $\tau_e$ time to traverse the arc, and therefore, they reach the head
 of $e$ earliest at time $\theta + \tau_e$. Hence, \eqref{eq:conservation_on_arc:multi_commodity} ensures that the
 amount of flow that leaves $e$ cannot exceed the amount of flow that has entered and traversed the arc. In other words,
 flow is not created within arcs. In addition, \eqref{eq:flow_conservation:multi_commodity} ensures that the network
 does not leak at intermediate nodes and that we have a constant network inflow rate of $\inrate_j$ at $\source_j$
 during $I_j$. Note that flow can be stored within arcs but never on nodes.
  
 \paragraph{Queues.} For every arc~$e$ there is a bottleneck given by its capacity~$\capa_e$ right before the head of
 the arc.  If more flow wants to leave~$e$ than the capacity allows, a queue builds up at the exit of the arc. The
 amount of total flow in the queue at time~$\theta$ is given by
 \[z_e(\theta) \coloneqq F_e^+(\theta - \tau_e) - F_e^-(\theta).\] 
 The queue does not have any physical dimension in the network and is therefore called \defemph{point queue}.
 
\paragraph{Feasibility.} 
A multi-commodity flow over time $f$ is \defemph{feasible} if the total outflow rate follows the same flow dynamic as for a feasible single-commodity flow over time, i.e., if
  \begin{equation}\label{eq:totaloutflow:multi_commodity}
  f_e^-(\theta) = \begin{cases}
  \capa_e & \text{ if } z_e(\theta) > 0, \\
  \min\Set{f_e^+(\theta - \tau_e), \capa_e} & \text{ if } z_e(\theta) = 0,
  \end{cases} \qquad \text{ for almost all } \theta \in \Time.
  \end{equation}
 
Furthermore, the amount of flow of a commodity $j$ that leaves an arc $e$ at time $\theta$ is determined by its fraction
of the total inflow rate at time $\vartheta$ when this particle entered the arc. In other words,
\begin{equation} \label{eq:FIFO:multi_commodity}
f^-_{j, e} (\theta) = \begin{cases}
f^-_e (\theta) \cdot \frac{f^+_{j, e}(\vartheta)}{f^+_e(\vartheta)} & \text{ if }  f^+_e(\vartheta) > 0,\\
0 & \text{ else,}
\end{cases}
\end{equation}
where $\vartheta = \min\set{\xi\leq \theta | T_e(\xi)=\theta}$ is the earliest point in time a particle can enter arc
$e$ in order to leave it at time $\theta$. This equation ensures that arcs preserve the proportion of commodities within
the flow as depicted in \Cref{fig:fifo:multi_commodity}. In particular, queues follow the first-in-first-out (FIFO)
principle, which means that particles cannot overtake others within the queues.

\begin{figure}[t]
\centering \includegraphics{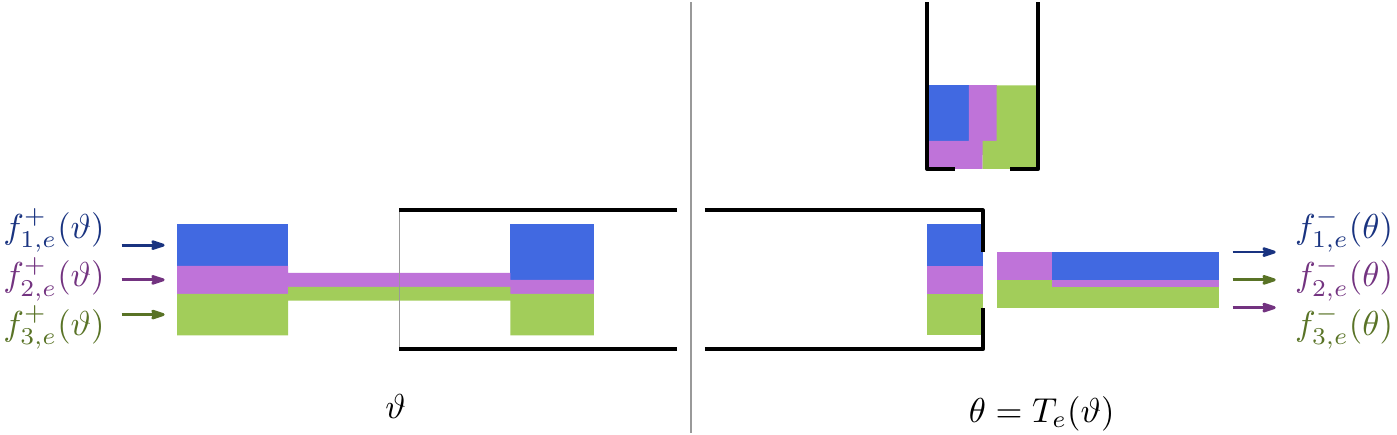} \caption{Inflow at time $\vartheta$ (\emph{left side}) and outflow at
time $\theta = T_e(\vartheta)$ (\emph{right side}) of three commodities. We require the flow to merge perfectly, which
means that the proportions of each commodity are conserved on an arc even if the flow is stretched or compressed.}
\label{fig:fifo:multi_commodity}
\end{figure}

\paragraph{Waiting times.} Given a feasible multi-commodity flow over time $f$, the \defemph{waiting
time}~$q_e \colon \Time \to \Time$ is given by
\[q_e(\theta) \coloneqq \frac{z_e(\theta+\tau_e)}{\capa_e}.\]
Note that $q_e(\theta)$ denotes the queue waiting time of particles that enter the \emph{arc} at time $\theta$, and therefore,
they enter the queue only at time $\theta + \tau_e$. Hence, the actual waiting period of those particles is given by
$[\theta + \tau_e, \theta + \tau_e + q_e(\theta)]$.

\paragraph{Exit times.} The \defemph{exit time} for arc $e$ is the function $\T_e \colon \Time \to \Time$ that maps the
entrance time~$\theta$ to the time a particle leaves the arc
\[\T_e(\theta) \coloneqq \theta + \tau_e + q_e(\theta).\]


Note that the total flow over time $(f_e^+, f_e^-)$ is a feasible single-commodity flow over time with respect to the
Koch-Skutella-model. Therefore, the following properties from the single-commodity model also holds for the total flow of a
multi-commodity flow over time.

\begin{restatable}{lemma}{technicalpropertiesbase} \label{lem:technical_properties:base}
    For a feasible flow over time $f$ it holds for all $e \in E$, $v \in V$ and~$\theta \in \Time$ that:
    \begin{enumerate}\leftskip=5mm
      \item $q_e(\theta) > 0 \;\; \Leftrightarrow \;\; z_e(\theta+\tau_e) > 0$. \label{it:q_equiv_z:base}
      \item $z_e(\theta +\tau_e + \xi) > 0$ \quad for all $\xi \in [0, q_e(\theta))$. \label{it:positive_queue_while_emptying:base}
      \item $F^+_e(\theta)=F^-_e(T_e(\theta))$. \label{it:in_equals_out_at_exit_time:base}
      \item For $\theta_1 < \theta_2$ with $F_e^+(\theta_2) - F_e^+(\theta_1) = 0$ and $z_e(\theta_2 + \tau_e)>0$  
      we have $T_e(\theta_1)=T_e(\theta_2)$. \label{it:equal_exit_times:base} 
      \item The functions $T_e$ are monotonically increasing. \label{it:T_monoton:base}
     \item The functions $F_e^+$, $F_e^-$, $z_e$, $q_e$ and $T_e$ are almost everywhere differentiable. \label{it:q_is_diffbar:base}
     \item For almost all $\theta \in \Time$ we have
     \[q'_e(\theta) = \begin{cases}
     \frac{f_e^+(\theta)}{\capa_e} - 1 & \text{ if } q_e(\theta) > 0,\\
     \max\Set{\frac{f_e^+(\theta)}{\capa_e} - 1, 0} & \text{ else.}
     \end{cases}\] \label{it:derivative_of_q:base}
    \end{enumerate}
\end{restatable}
Most of the statements follow immediately from the definitions and some involve some minor calculations. For
\ref{it:q_is_diffbar:base} we use Lebesgue's differentiation theorem. As the proof does not give any interesting further insights we moved it to the
appendix on page \pageref{proof:technical_properties:base}.
  
Additionally, the follow property holds for every commodity
separately.

\begin{restatable}{lemma}{fifomulticommodity}\label{lem:fifo:multi_commodity}
For a feasible multi-commodity flow over time $f$ we have for every arc $e \in E$, every commodity $j \in \J$ and all
$\theta \in \Time$ that
\[F_{j, e}^+(\theta) = F_{j, e}^-(\T_e(\theta)).\]
\end{restatable}

\proof{Proof of \Cref{lem:fifo:multi_commodity}.}
By \Cref{lem:technical_properties:base}~\ref{it:in_equals_out_at_exit_time:base} we have that $F_e^+(\xi) =
F_e^-(\T_e(\xi))$. Taking the derivative yields that $f_e^+(\xi) = f_e^-(\T_e(\xi)) \cdot \T'_e(\xi)$ for almost all
$\xi \in [0, \theta]$. Hence, for $f_e^+(\xi) > 0$ we obtain
\begin{equation}\label{eq:derivative_of_cumulative_commodity_flow:multi_commodity}
\frac{\diff}{\diff \xi} F_{j, e}^-(\T_e(\xi)) =  f_{j, e}^-(\T_e(\xi)) \cdot \T'_e(\xi)
\!\stackrel{\text{\eqref{eq:FIFO:multi_commodity}}}{=}\! f^-_e (\T_e(\xi)) \cdot \frac{f^+_{j, e}(\xi)}{f^+_e(\xi)} 
\cdot \T'_e(\xi)= f_{j, e}^+(\xi).\end{equation}
In the case of $f_e^+(\xi) = 0$ both sides equal $0$.
Taking the integral over $[0, \theta]$ of \eqref{eq:derivative_of_cumulative_commodity_flow:multi_commodity} yields
$F_{j, e}^-(\T_e(\theta)) = F_{j, e}^+(\theta)$ since $F_{j,e}^-(T_e(0)) = F_{j,e}^+(0) = 0$.
\hfill\Halmos\endproof

Note that Condition \eqref{eq:conservation_on_arc:multi_commodity} is not used in the proof. Since $F_{j, e}^+(\theta) = F_{j,
e}^-(\T_e(\theta))$ implies flow conservation on arcs, we can again drop this condition for a feasible multi-commodity flow over time.

\section{Multi-commodity Nash flows over time.} \label{sec:nash_flows:multi-commodity} In order to define dynamic equilibria we first have to
transfer the concept of current shortest paths networks and resetting arcs from \cite{koch2011nash} to the multi-commodity case.

\paragraph{Earliest arrival times.} Since every flow commodity has its own origin we need to define earliest arrival
time functions for every commodity separately. For a given flow over time $f$ let $\l_{j,v} \colon \R \to [0, \infty)$
be the earliest time a particle of commodity $j$ can arrive at $v$. More precisely, we define the \defemph{earliest
arrival time} for commodity $j \in \J$ by
\begin{align}  \label{eq:bellman:multi_commodity}
\begin{aligned}
\l_{j,\source_j}(\phi) &\coloneqq \frac{\phi}{r_j} + a_j,&&\\ 
\l_{j,v}(\phi) &\coloneqq \min_{e = u v\in E} \T_e(\l_{j,u}(\phi))  &&\quad \text{ for } v
\in V\backslash \set{\source_j}.
\end{aligned}
\end{align}

\new{The flow of a commodity can be seen as an infinite long area of width $1$, which means that the flow volume of an
interval $[a, b] \subset \R$ equals $b - a$ (more general: the volume of a measurable subset of $\R$ is given by its
Lebesgue-measure). Furthermore, only the particles in $\K_j \coloneqq [0, (b_j - a_j) \cdot r_j)$ enters the network within the time
interval $I_j$. For technical reasons we also define the earliest arrival times for particles $\phi \not\in \K_j$ by setting $q_e(\theta) = 0$ for all $\theta < 0$. This way the earliest arrival time functions are
surjective on $\R$.}

\paragraph{Current shortest paths networks and active arcs.}
In an equilibrium every particle wants to get to its destination
as fast as possible and will therefore use a shortest path.
We say an arc $e = uv$ is \defemph{active} for particle $\phi$ and commodity~$j$
if $\l_{j,v}(\phi) = \T_e(\l_{j,u}(\phi))$ and we denote the set of all active arcs for $\phi$ and $j$ by
\[E'_{j, \phi} \coloneqq \set{ e = uv \in E | \l_{j, v}(\phi) = \l_{j, u}(\phi) + \tau_e + q_e(\l_u(\phi))}.\]
The graph $G_{j, \phi} \coloneqq (V, E'_{j, \phi})$ is called the \defemph{current shortest paths network} of
particle $\phi$ and commodity~$j$. 

\paragraph{Resetting arcs.} It will be important to specify the arcs at which a particle would experience a waiting time
when traveling along a shortest path. Hence, we define
\[E^*_{j, \phi} \coloneqq \Set{e = uv \in E |q_e(\l_{j,u}(\phi)) > 0}\]
to be the \defemph{resetting} arcs for particle $\phi$ and commodity~$j$.
Note that there might be arcs that are resetting but not active for some commodity.

\paragraph{Dynamic equilibria.}
Since every particle wants to arrive at its destination as early as possible, it should
only use current shortest paths, which leads to the following definition.
\begin{definition}[Multi-commodity Nash flow over time] \label{def:Nash_flow:multi_commodity} 
A feasible multi-commodity flow over time $f$ is a \defemph{multi-commodity Nash flow over time} if 
\begin{equation}\label{eq:Nash_condition:multi_commodity}\tag{N}
f_{j, e}^+(\theta) > 0 \;\;\Rightarrow\;\; \theta \in \l_{j,u}(\Phi_{j,e}) \qquad \text{ for all } 
e = uv \in E, j \in \J \text{ and almost all } \theta \in [0, \infty),\end{equation}
where $\Phi_{j,e} \coloneqq \set{\phi \in \Flow | e \in E'_{j, \phi}}$ is the set of flow particles of commodity $j$ for
which arc $e$ is active.
\end{definition}

We can characterize Nash flows over time in the multi-commodity setting as follows.
\begin{restatable}{lemma}{nashflowcharacterizationmulticommdoity}\label{lem:nash_flow_characterization:multi_commodity}
Let $f$ be a feasible multi-commodity flow over time. The following statements are equivalent:
\begin{enumerate}\leftskip=5mm
  \item $f$ is a multi-commodity Nash flow over time. \label{it:nash_flow:multi_commodity}
  \item $F^+_{j, e}(\l_{j, u}(\phi)) = F_{j, e}^-(\l_{j, v}(\phi))$ for all $e = uv$, $j \in \J$ and all~$\phi \in \Flow$.
  \label{it:in_equals_out_at_l:multi_commodity}
\end{enumerate}
\end{restatable}
The
intuitive idea of the proof is the following. Either arc $e$ is active for commodity $j$, then the equation
follows from \Cref{lem:fifo:multi_commodity} or $e$ is not active for commodity $j$, but
then the Nash condition \eqref{eq:Nash_condition:multi_commodity} states that there was no inflow of commodity $j$ between the last point in time $\theta$ when this arc was active and $\l_{j, u}(\phi)$. Hence, we have $F^+_{j, e}(\l_{j, u}(\phi)) = F^+_{j, e}(\theta) = F^-_{j, e}(T_e(\theta)) \leq F^-_{j, e}(\l_v(\phi))$, which
together with the flow conservation on arcs \eqref{eq:conservation_on_arc:multi_commodity} shows that $F^+_{j, e}(\l_{j, u}(\phi)) =
F^-_{j, e}(\l_{j, v}(\phi))$. A detailed version of the proof can be found in the appendix
on page \pageref{proof:nash_flow_chracterization:multi_commodity}.

In a multi-commodity Nash flow over time we can characterize the waiting times, and therefore the active and resetting
arcs, by the earliest arrival time functions alone, which is shown in the following lemma.

\begin{lemma} \label{lem:q_characterized_by_l:multi_commodity}
Given a multi-commodity Nash flow over time $f$ with arrival time functions $(\l_{j, v})_{j \in \J, v \in V}$,
 we have for all arcs $e = uv \in E$ and all $\theta \in \Time$ that 
 \[q_e(\theta) = \max_{j \in \J}\left( \max\Set{ \l_{j, v}(\phi_j) - \l_{j, u}(\phi_j) - \tau_e, 0}\right) 
 \qquad \text{ with } \quad \phi_j \coloneqq \min \set{\phi \in \Flow | \l_{j, u}(\phi) = \theta}.\]
\end{lemma}

\proof{Proof of \Cref{lem:q_characterized_by_l:multi_commodity}.}
If $q_e(\theta) = 0$ we have by \eqref{eq:bellman:multi_commodity} for all commodities $j$ that
\[\l_{j, v}(\phi_j) \leq \l_{j, u}(\phi_j) + \tau_e + q_e(\l_{j, u}(\phi_j)) = \l_{j, u}(\phi_j) + \tau_e.\] 
For $q_e(\theta) > 0$ we show that there has to be at least one commodity $j \in \J$ for which $e$ is active for
particle $\phi_j$. Let $j$ be the commodity for which $e$ was active at the latest point in time before $\theta$, i.e.,
\[j \coloneqq \argmax_{i \in \J} \l_{i, u}(\varphi_{i}) \qquad \text{ with } \quad \varphi_{i} 
\coloneqq \max\set{\xi \leq \phi_j | e \in E'_{i, \xi}}.\] 
Since no flow was sent into $e$ between $\l_{j, u}(\varphi_j)$ and $\theta = \l_{j, u}(\phi_{j, u})$ we obtain for the
total cumulative inflow that $F_e^+(\l_{j, u}(\phi_j)) - F_e^+(\l_{j, u}(\varphi_j)) = 0$. Hence, by
\Cref{lem:technical_properties:base}~\ref{it:equal_exit_times:base} we get
\[\l_{j, v}(\phi_j) \leq T_e(\l_{j, u}(\phi_j)) = T_e(\l_{j, u}(\varphi_j)) = \l_{j, v}(\varphi_j) \leq \l_{j, v}(\phi_j).\]
Thus, we have equality, which shows that $e$ is active for $\phi_j$. It follows that
\[q_e(\theta) = q_e(\l_{j, u}(\phi_j)) = \l_{j, v}(\phi_j) - \l_{j, u}(\phi_j) - \tau_e.\]
Clearly, there cannot be another commodity $j'$ with $\l_{j', v}(\phi_{j'}) - \l_{j', u}(\phi_{j'}) - \tau_e >
q_e(\theta)$ since this would contradict the definition of the earliest arrival times in
\Cref{eq:bellman:multi_commodity}.
\hfill\Halmos\endproof

\paragraph{Underlying static flows.} We define the \defemph{underlying static flow} for each commodity $j$ by
\[x_{j, e}(\phi) \coloneqq F_{j, e}^+(\l_{j,u}(\phi)) = F_{i, e}^-(\l_{i,v}(\phi)) \qquad \text{ for all } e = uv \in E.\] 
It is easy to see that the arc vector~$(x_j(\phi))_{e \in E}$ forms a static $\source_j$-$\sink_j$-flow of value $\phi$
if $\phi \in \K_j$. Furthermore, these functions are monotone and almost everywhere
differentiable and the vector of derivatives~$(x'_j(\phi))_{e \in E}$ forms a static $\source_j$-$\sink_j$-flow of value
$1$ for $\phi \in \K_j$ or of value $0$ otherwise.

\paragraph{Underlying foreign flows.} The main challenge of multi-commodity dynamic equilibria is that the stress of an
arc, and therefore the route choice of each particle, depends on flow of all commodities simultaneously. To obtain some
structural insight we define the \defemph{underlying foreign flow}~by
\[y_{j, e}(\phi) \coloneqq \sum_{i \in \J \backslash \set{j}} F_{i, e}^+(\l_{j,u}(\phi)).\]
Note that this is not a static flow in general since the cumulative inflow $F_{i, e}^+(\l_{j,u}(\phi))$ of some
commodity~$i$ into an arc $e$ that is active for commodity~$i$ but not for commodity~$j$ generally differs from the
cumulative outflow $F_{i, e}^-(\l_{j,v}(\phi))$.

Nonetheless, we have
\[y_{j, e}(\phi) = \sum_{i \in \J \backslash \set{j}} x_{i, e}(\phi^i_{j, u}) \qquad 
\text{ with }\quad \phi_{j, u}^i \coloneqq \min \l_{i, u}^{-1}(\l_{j,u}(\phi)).\]

Note that $\phi^i_{j, u}$ is the very first particle of commodity $i$ that can arrive at $u$ (when taking a shortest
path) exactly at the time when the particle $\phi$ of commodity $j$ reaches $u$. It is, therefore, a function in
dependency of $\phi$, but for sake of readability we omit the parameter in most cases.

\begin{lemma}\label{lem:derivative_of_foreign_flow:multi_commodity}
For all $j \in \J$ and $e \in E$ the underlying foreign flow $y_{j,e}(\phi)$ is almost everywhere differentiable with
\[y'_{j, e}(\phi)  = \sum_{i \in \J \setminus \set{j}} f_{i, u}^+(\l_{j, u}(\phi)) \cdot \l'_{j, u}(\phi) =
\begin{cases}
\sum_{i \in \J \setminus \set{j}} x'_{i, e}(\phi^i_{j, u}) \cdot \frac{\l'_{j,u} (\phi)}{\l'_{i, u}(\phi^i_{j, u})} 
& \text{ if } \l'_{j,u}(\phi) > 0,\\
0 & \text{ else.}
\end{cases}\]
\end{lemma}
An illustration of the relation between the foreign inflow rates and the derivatives of the underlying foreign flow can
be found in \Cref{fig:foreign_flow:multi_commodity}.

\begin{figure}[t]
\centering \includegraphics{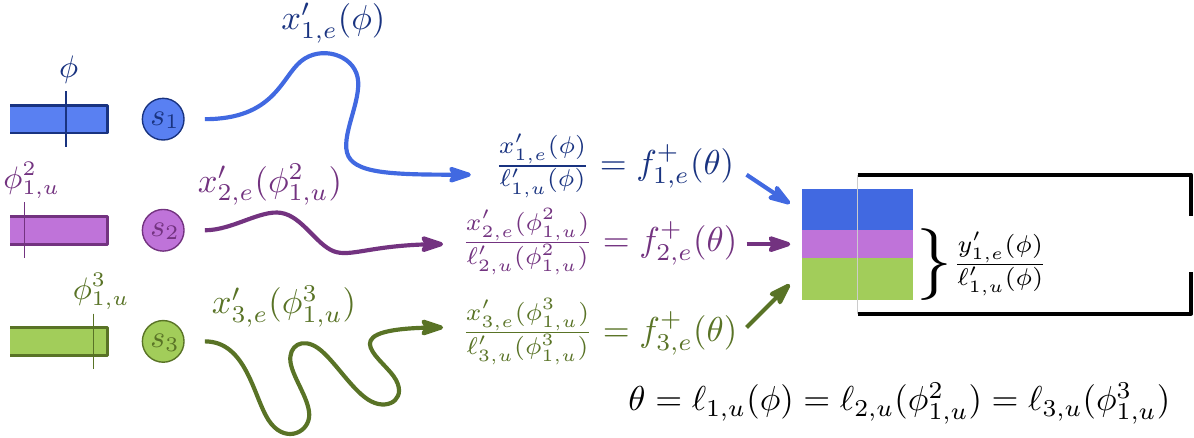} \caption{Foreign flow entering an arc. Particle $\phi$ of
commodity $1$ is entering arc $e = uv$ at time $\theta \coloneqq \l_{1, u}(\phi)$. To determine the inflow rate of the
other commodities at this point in time, we consider the particles $\phi_{1, u}^i$ that also reach node~$u$ at
time~$\theta$. The value $x'_{i, e}(\phi_{1, u}^2)$ denotes the part of the flow of commodity $i$ that will use arc~$e$.
Hence, we obtain the foreign inflow rates at time $\theta$ by dividing this value by $\l'_{i, u}(\phi_{1, u}^i)$.}
\label{fig:foreign_flow:multi_commodity}
\end{figure}

\proof{Proof of \Cref{lem:derivative_of_foreign_flow:multi_commodity}.}
By Lebesgue's theorem for the differentiability of monotone functions, $\phi^i_{j, u}(\phi)$ is almost everywhere
differentiable as it is monotone. As a concatenation and sum of almost everywhere differentiable functions so is
$y_{j,e}(\phi)$. Let $\phi$ be a particle such that the functions $\l_{j, u}(\phi)$, $\phi^i_{j. u}(\phi)$, $\l_{i,
u}(\phi^i_{j, u}(\phi))$ and $y_{j,e}(\phi)$ are differentiable for all $i\in \J\setminus \set{j}$. This is given for
almost all particles. The first equation follows immediately by the chain rule. For $\l'_{j, u}(\phi) = 0$ we have
$y'_{j, e}(\phi) = 0$. So let us suppose that $\l'_{j, u}(\phi) > 0$. We obtain
\[0 < \l'_{j, u}(\phi) = \frac{\diff}{\diff \phi} \l_{i, u} (\phi^i_{j, u}(\phi)) 
= \l'_{i, u}(\phi^i_{j, u}(\phi)) \cdot \frac{\diff}{\diff \phi} \phi^i_{j, u}(\phi),\]
and therefore, $\l'_{i, u}(\phi^i_{j, u}(\phi)) > 0$.
Again with the chain rule and the equation above it follows immediately that
\[y'_{j, e}(\phi) = \sum_{i \in \J \setminus \set{j}} x'_{i, e}(\phi^i_{j, u}(\phi)) 
\cdot \frac{\diff}{\diff \phi} \phi^i_{j, u}(\phi)  
=  \sum_{i \in \J \setminus \set{j}} x'_{i, e}(\phi^i_{j, u}(\phi)) 
\cdot \frac{\l'_{j,u} (\phi)}{\l'_{i, u}(\phi^i_{j, u}(\phi))}.\]
\hfill\Halmos\endproof

\section{Multi-commodity thin flows.} \label{sec:thin_flows:multi_commodity} 
In the considered flow over time game every particle $\phi$ of commodity $j$
does not only choose one $\source_j$-$\sink_j$-path but it can also split up even further and each part can take a different
path from $\source_j$ to $\sink_j$. Hence, a strategy of this particle is, in fact, a convex combination of such paths,
or in other words, a strategy is given by a static $\source_j$-$\sink_j$-flow of value~$1$.

It turns out that for Nash flows over time the strategies are given by the derivatives of the underlying static flows. In order to describe the structure
of these derivatives we extend the thin flow definition of \cite{koch2011nash} to the multi-commodity setting.
However, we have to include the derivatives of the foreign flow into our consideration, and since the foreign flow
heavily depends on the underlying static flow of other commodities we cannot consider only one particle (or one interval
of particles) at a time, but we have to consider the strategy of all particles simultaneously.

\begin{definition}[Multi-commodity thin flow]
For a given family of arc functions $x' = (x'_{j, e})_{j \in \J, e \in E}$ and node functions $\l' = (\l'_{j, v})_{j \in
\J, v \in V}$ we define $E'_{j, \phi}, E^*_{j, \phi} \subseteq E$ and $y'_{j, e} \colon \Flow \to [0, \infty)$ for all
$j \in \J$, $\phi \in \Flow$ and $e \in E$ as described above in dependency of the functions $\l_{j,v}(\phi) \coloneqq
\int_0^\phi \l'_{j, v}(\xi) \diff \xi$. We say that the pair~$(x', \l')$ forms a \defemph{multi-commodity thin flow} if
the following conditions are satisfied:

For all $\phi \in \Flow$ the arc vector $(x'_{j, e}(\phi))_{e \in E}$ forms a static $\source_j$-$\sink_j$-flow of value
$1$ if $\phi \in \K_j$ or of value $0$ if $\phi \notin \K_j$. In both cases we have $x'_{j, e}(\phi) = 0$ for all $e
\notin E'_{j, \phi}$ and for almost all $\phi \in \Flow$ the following equations hold:
\begin{align}
\l'_{j, \source_j}(\phi) &= \frac{1}{r_j} &&\text{ for all } j \in \J, \label{eq:l'_s:multi_commodity} \tag{TF1}\\ 
\l'_{j,v}(\phi) &= \min_{e = uv \in E'_{j, \phi}} \rho_{j, e}^\phi\left(\l'_{j,u}(\phi), x'_{j, e}(\phi), y'_{j,e}(\phi)
\right) && \text{ for all } j \in \J, v \in V \setminus \Set{\source_j},
\label{eq:l'_v_min:multi_commodity}\tag{TF2}\\ 
\l'_{j,v}(\phi) &= \rho_{j, e}^\phi\left(\l'_{j,u}(\phi), x'_{j, e}(\phi), y'_{j, e}(\phi)\right) &&\!\!\! \begin{array}{ll}
\text{ for all } j \in \J, e = uv \in E'_{j, \phi}\\
 \text{ with } x'_{j, e} > 0,
\end{array} \label{eq:l'_v_tight:multi_commodity}\tag{TF3}
\end{align}
\begin{align*}\text{ where } \qquad \rho_{j, e}^\phi(\l'_{j, u}, x'_{j, e}, y'_{j, e}) &\coloneqq \begin{cases}
\frac{x'_{j, e} + y'_{j, e}}{\capa_e} & \text{ if } e = uv \in E^*_{j, \phi},\\
\max\Set{\l'_{j, u}, \frac{x'_{j, e} + y'_{j, e}}{\capa_e}} & \text{ if } e = uv \in E'_{j, \phi} \backslash E^*_{j, \phi}. 
\end{cases}
\end{align*}
\end{definition}

The intuitive idea behind these equations is that $\frac{x'_{j, e} + y'_{j, e}}{\capa_e}$ describes the \defemph{stress value} of an arc $e$ in dependency of the particles using this arc. The stress value of an $\source_j$-$v$-path is then
determined by the highest stress value of its arcs (i.e., by the bottleneck arc along the path), as long as there are no
resetting arcs; see $\rho^\phi_{j.e}(l'_{j, u}, x'_{j, e}, y'_{j, e})$ if $e \notin E_{j, \phi}^*$. At every resetting arc along the path the values of all
previous arcs are dismissed and the stress value of the path is reset; see $\rho^\phi_{j.e}(l'_{j, u}, x'_{j, e}, y'_{j, e})$ if $e \in E_{j, \phi}^*$. This
is logical since a high stress value of preceding arcs can be compensated by decreasing the queue, as long as there is
a positive queue.

Finally, the value $\l'_{j, v}$ is the minimal stress value of all paths from $\source_j$ to $v$; see
\eqref{eq:l'_v_min:multi_commodity}. For high stress values following particles need more time to reach node $v$, hence, these
stress values coincide exactly with the slope of the earliest arrival time functions.

If $\l'_{j, v} < \rho^\phi_{j.e}(l'_{j, u}, x'_{j, e}, y'_{j, e})$, this means that $e$ leaves the current shortest paths network, and therefore, it cannot
be used by following particles, i.e., $x'_{j,e} = 0$. In other words, particles in a Nash flow over time
can only use arcs that lie on a path with minimal stress value; see \eqref{eq:l'_v_tight:multi_commodity}.

Furthermore, $\l'_{j, u} < \frac{x'_{j, e} + y'_{j, e}}{\capa_e}$ means that arc $e$ is a bottleneck, and therefore, the queue will grow.
Whenever we have $\l'_{j, u} > \frac{x'_{j, e} + y'_{j, e}}{\capa_e}$ the arc $e$ has a smaller stress value than the preceding arcs along the
$\source_j$-$v$-path. Hence, the queue will decrease if $e$ is resetting, or stay empty otherwise. For $\l'_{j, u} =
\frac{x'_{j, e} + y'_{j, e}}{\capa_e}$ the arc has the exact same stress value as the arcs before, so the queue will stay
constant.

The first main result for multi-commodity Nash flows over time states that the derivatives form a multi-commodity thin flow.

\begin{restatable}{theorem}{thmderivativesformthinflowsmulticommodity}\label{thm:derivatives_form_thin_flows:multi_commodity}
For a multi-commodity Nash flow over time $f$, the derivatives $(x'_{j, e})_{j \in J, e \in E}$ together with
$(\l'_{j, v})_{j \in \J, v \in V}$ form a multi-commodity thin flow.
\end{restatable}
\proof{Proof of \Cref{thm:derivatives_form_thin_flows:multi_commodity}.} \label{proof:derivates_form_thin_flows:multi_commodity}
Let $\phi \in \Flow$ be a particle such that for all arcs $e = uv$ and all $j \in \J$ the derivatives of $\l_{j,u}$,
$x_{j, e}$, $y_{j, e}$ and $T_e \circ \l_{j, u}$ exist. Furthermore, assume that 
\[x'_{j. e}(\phi) = f_{j, e}^+(\l_{j,u}(\phi)) \cdot \l'_{j, u}(\phi) = f_{j, e}^-(\l_{j, v}(\phi)) \cdot \l'_{j, v}(\phi)\]
and that
\eqref{eq:Nash_condition:multi_commodity} as well as the equation in
\Cref{lem:derivative_of_foreign_flow:multi_commodity} hold. This is given for almost all $\phi \in \Flow$.
By \eqref{eq:Nash_condition:multi_commodity} we have $f_{j, e}^+(\phi) = 0$, and therefore $x'_{i, e}(\phi) = 0$, for
all arcs $e \in E \setminus E'_{i, \phi}$, which shows that $(x'_{j, e})_{e \in E}$ is indeed a static flow on $G_{j,
\phi}$.

Taking the derivatives of the first equation of \eqref{eq:bellman:multi_commodity} shows immediately
\eqref{eq:l'_s:multi_commodity}.

In oder to show \eqref{eq:l'_v_min:multi_commodity} we add $f^+_{j, u}(\l_{j, u}(\phi)) \cdot
\l'_{j,u}(\phi) = x'_{j,e}(\phi)$ to the equation in \Cref{lem:derivative_of_foreign_flow:multi_commodity} and obtain
\[f_e^+(\l_{j, u}(\phi)) \cdot \l'_{j,u}(\phi) = \sum_{i \in \J} f_{i, e}^+(\l_{j, u}(\phi)) \cdot \l'_{j,u}(\phi) 
= x'_{j, e}(\phi) + y'_{j, e}(\phi).\]
Furthermore, by \Cref{lem:technical_properties:base}~\ref{it:derivative_of_q:base} we have for almost all $\theta \in
\Time$ that
\[\T'_e(\theta) = 1 + q'_e(\theta) = \begin{cases}
\max\Set{\frac{f_e^+(\theta)}{\capa_e}, 1} & \text{ if } q_e(\theta) = 0,\\
\frac{f_e^+(\theta)}{\capa_e} & \text{ else.}
\end{cases}\]
Hence,
\[\frac{\diff}{\diff \phi} \T_e(\l_{j,u}(\phi)) = \T'_e(\l_{j,u}(\phi)) \cdot \l'_{j,u}(\phi) = \begin{cases}
\max\Set{\frac{x'_{j, e}(\phi) + y'_{j, e}(\phi)}{\capa_e}, \l'_{j,u}(\phi)}& \text{ if } q_e(\l_{j, u}(\phi)) = 0,\\
\frac{x'_{j, e}(\phi) + y'_{j, e}(\phi)}{\capa_e} &\text{ else.} \end{cases} \] 
This, together with \eqref{eq:bellman:multi_commodity} and the differentiation rule for a minimum
(see \Cref{lem:diff_rule_for_min:pre}), implies \eqref{eq:l'_v_min:multi_commodity}.
 
In oder to prove \eqref{eq:l'_v_tight:multi_commodity} suppose $x'_{j, e}(\phi) = f_{j,
e}^-(\l_v(\phi)) \cdot \l'_{j,v}(\phi) > 0$, which implies $f_e^+(\l_{j, u}(\phi)) \geq f_{j,e}^+(\l_{j, u}(\phi)) > 0$.
Since $e$ is active for $j$ we have $\l_{j,v}(\phi) = T_e(\l_{j, u}(\phi))$. Hence,

\begin{align*}
\l'_{j,v}(\phi) &= \frac{x'_{j, e}(\phi)}{f_{j, e}^-(\l_{j, v}(\phi))} \\
&\hspace{-1.3mm}\stackrel{\eqref{eq:FIFO:multi_commodity}}{=}\hspace{-1.3mm} 
\frac{x'_{j, e}(\phi) \cdot f_e^+(\l_{j,u}(\phi))}{f_{j, e}^+(\l_{j,u}(\phi)) \cdot f_e^-(\l_{j, v}(\phi))} \\
&= \frac{\l'_{j, u}(\phi) \cdot f_e^+(\l_{j,u}(\phi))}{f_e^-(\l_{j, v}(\phi))}\\
&\hspace{-1.3mm}\stackrel{\eqref{eq:totaloutflow:multi_commodity}}{=}\hspace{-1.3mm} \begin{dcases}
\max \Set{\l'_{j, u}(\phi), \frac{\l'_{j, u}(\phi) \cdot f_e^+(\l_{j, u}(\phi))}{\capa_e}} 
&\text{ if } q_e(\l_{j, u}(\phi)) = 0,\\
\frac{\l'_{j, u}(\phi) \cdot f_e^+(\l_{j, u}(\phi))}{\capa_e} &\text{ else},
\end{dcases}\\
&=\begin{dcases}
\max\Set{\l'_{j, u}(\phi), \frac{x'_{j, e}(\phi) + y'_{j, e}(\phi)}{\capa_e}} 
&\text{ if } e \in E'_{j, \phi} \backslash E^*_{j, \phi},\\
\frac{x'_{j, e}(\phi) + y'_{j, e}(\phi)}{\capa_e} &\text{ if } e \in E^*_{j, \phi},
\end{dcases}\\
&= \rho_{j, e}^\phi\left(\l'_{j, u}(\phi), x'_{j, e}(\phi), y'_{j, e}(\phi)\right).
\end{align*}
Thus, \eqref{eq:l'_v_tight:multi_commodity} is fulfilled, which finishes the proof.
\hfill\Halmos\endproof

For the reverse direction we show that for a given multi-commodity thin flow $(x', \l')$ we can reconstruct the Nash
flow over time by setting
\begin{align*}f_{j, e}^+(\theta) \coloneqq \frac{x'_{j, e}(\phi)}{\l'_{j, u}(\phi)} \quad \text{ for } \theta 
= \l_{j ,u}(\phi) \qquad \text{ and } \qquad
f_{j, e}^-(\theta) \coloneqq \frac{x'_{j, e}(\phi)}{\l'_{j, v}(\phi)} \quad \text{ for } \theta = \l_{j ,v}(\phi)
\end{align*}
for all $\phi \in \Flow$ and every $e = uv \in E$. Furthermore, we set $f_{j, e}^+(\theta) = 0$ for $\theta < \l_{j,
u}(0)$ and $f_{j, e}^-(\theta) = 0$ for $\theta < \l_{j, v}(0)$.

\begin{theorem} \label{thm:thin_flows_form_Nash_flows:multi_commodity}
For every multi-commodity thin flow $(x', \l')$ the family of functions $f = (f_{j, e}^+, f_{j, e}^-)_{j \in \J, e \in
E}$ as defined above is a multi-commodity Nash flow over time with earliest arrival time functions
\[\l_{j,v}(\phi) \coloneqq \int_0^\phi \l'_{j, v}(\xi) \diff \xi \qquad 
\text{for all } j \in \J, v \in V \text{ and } \phi \in \Flow.\]
\end{theorem}
\proof{Proof of \Cref{thm:thin_flows_form_Nash_flows:multi_commodity}.}
Clearly, \eqref{eq:flow_conservation:multi_commodity} is satisfied since flow of every commodity $j$ is conserved at
every node $v$ at every point in time $\theta = \l_{j, v}(\phi)$, i.e.,
\[ \sum_{e\in \delta^+_v} f_{j, e}^+(\theta) - \sum_{e \in \delta^-_v} f_{j, e}^-(\theta) 
= \sum_{e\in \delta^+_v} \frac{x'_{j, e}(\phi)}{\l'_{j, v}(\phi)} 
- \sum_{e \in \delta^-_v} \frac{x'_{j, e}(\phi)}{\l'_{j, v}(\phi)}  
= \begin{cases}
0 & \text{ if } v \in V\setminus \set{\source_j} \text{ or } \phi \notin \K_j, \\
r_j & \text{ if } v = \source_j \text{ and } \phi \in \K_j.
\end{cases} \]
Note that for $v = \source_j$ we have $\phi \in \K_j$ if, and only if, $\theta = \l_{j, \source_j}(\phi) \in I_j$.
 
For a given $e = uv \in E$ and $\theta \in \Time$ let $\phi_j \in \Flow$ such that $\l_{j, u}(\phi_j) = \theta$ for all
$j \in \J$. Considering the commodities $j$, where $e$ is active for $j$ and $\phi_j$, we observe that also all $\l_{j,
v}(\phi_j)$ of these commodities coincide. Hence, \eqref{eq:l'_v_tight:multi_commodity} yields
\begin{align*}f_e^-(\theta) = \sum_{j \in \J} \frac{x'_{j ,e}(\phi_j)}{\l'_{j, v}(\phi_j)} 
&= \begin{cases} \capa_e & \text{ if } e \in E^*_{j, \phi_j} \text{ for some $j$ with $e \in E'_{j, \phi_j}$}, \\
\min \Set{\sum_{j \in \J} \frac{x'_{j ,e}(\phi_j)}{\l'_{j, u}(\phi_j)}, \capa_e} 
& \text{ else,}\end{cases}\\
&= \begin{cases} \capa_e & \text{ if } q_e(\theta) > 0, \\
\min \Set{f_e^+(\theta), \capa_e} 
& \text{ else.}\end{cases}\end{align*}
This shows \eqref{eq:totaloutflow:multi_commodity}.

\Cref{eq:FIFO:multi_commodity} follows by \Cref{lem:derivative_of_foreign_flow:multi_commodity} since
\[f^-_{j, e}(\theta) = \frac{x'_{j, e}(\phi_j)}{\l'_{j, v}(\phi_j)} 
= \frac{y'_{j, e}(\phi_j) + x'_{j, e}(\phi_j)}{\l'_{j, v}(\phi_j)} \cdot \frac{x'_{j, e}(\phi_j)}{\l'_{j, u}(\phi_j)} 
\cdot \frac{\l'_{j, u}(\phi_j)}{y'_{j, e}(\phi_j) + x'_{j, e}(\phi_j)} \\
= f^-_e(\theta) \cdot \frac{f^+_{i, e}(\theta)}{f_e^+(\theta)}.\]
 
In order to show that the $\l$-functions satisfy \Cref{eq:bellman:multi_commodity} we prove that the derivatives of
$\l_{j, v}(\phi)$ and of $\min_{e = uv \in E} T_e(\l_{j, u}(\phi))$ coincide for all $\phi \in \Flow$.
\Cref{lem:technical_properties:base}~\ref{it:derivative_of_q:base} implies for almost all $\theta \in \Time$ that
\[\T'_e(\theta) = 1 + q'_e(\theta) = \begin{cases}
\max\Set{\frac{f_e^+(\theta)}{\capa_e}, 1} & \text{ if } q_e(\theta) = 0,\\
\frac{f_e^+(\theta)}{\capa_e} & \text{ else.}
\end{cases}\]
Hence, \begin{align*}\T'_e(\l_{j,u}(\phi)) \cdot \l'_{j,u}(\phi) =
\begin{cases}
\max\Set{\frac{x'_{j, e}(\phi) + y'_{j, e}(\phi)}{\capa_e}, \l'_{j,u}(\phi)}& \text{ if } q_e(\l_{j, u}(\phi)) = 0,\\
\frac{x'_{j, e}(\phi) + y'_{j, e}(\phi)}{\capa_e} &\text{ else.}
\end{cases} \end{align*}
This together with \eqref{eq:l'_v_min:multi_commodity} and the differentiation rule for a minimum
(\Cref{lem:diff_rule_for_min:pre}) implies that
\[\l'_{j, v}(\phi) = \frac{\diff}{\diff \phi} \min_{e = uv \in E} T_e(\l_{j, u}(\phi)).\]
By Lebesgue's differentiation theorem we obtain \eqref{eq:bellman:multi_commodity}. In other words, the $\l$-functions
are indeed the earliest arrival times for the constructed feasible flow over time $f$.

Finally, $f$ is a multi-commodity Nash flow over time by \Cref{lem:nash_flow_characterization:multi_commodity} since
\[F_{j, e}^+(\l_{j, u}(\phi)) = \!\int_0^\phi f^+_{j, e}(\l_{j, u}(\xi)) \cdot \l'_{j, u}(\xi) \diff \xi 
= \!\int_0^\phi x'_{j, e} (\xi) \diff \xi =\! \int_0^\phi f^-_{j, e}(\l_{j, v}(\xi)) \cdot \l'_{j, v}(\xi) \diff \xi
= F_{j, e}^-(\l_{j, v}(\phi))\]
for all $e = uv \in E$, $j \in \J$ and $\phi \in \Flow$.
\hfill\Halmos\endproof

To sum this up, \Cref{thm:derivatives_form_thin_flows:multi_commodity,thm:thin_flows_form_Nash_flows:multi_commodity}
show that multi-commodity Nash flows over time correspond one-to-one to multi-commodity thin flows.

\section{Existence of multi-commodity Nash flows over time.}
\label{sec:existence_of_multi_commodity_nash_flows:multi} The existence of dynamic equilibria in a
multi-commodity setting was first shown by Cominetti, Correa and Larr\'e in \cite{cominetti2015dynamic}, even though the
proof is not worked out in this paper. The main idea of their proof is to represent feasible multi-commodity flows over time in a
path-based formulation, i.e., as vectors of inflow functions in the $L^p$-space, and then to formulate the Nash flow
condition as an infinite-dimensional variational inequality. Using by Br\'ezis' theorem (see \Cref{thm:brezis:prelim}) guarantees the existence of a multi-commodity Nash flow over time. 
\new{In order to maintain a better structural understanding in our existence proof we stick with the arc-based representations for flows over time and thin flows. The key idea is to utilize variational inequalities, and in particular Br\'ezis' theorem, to show the existence of a
multi-commodity thin flow.}

In order to avoid degenerated cases we assume from now on that all transit times are strictly positive. Let $H > 0$ such
that $\K_j \subseteq [0, H]$ for all $j \in \J$. We consider a vector of functions $x' =
(x'_{j,e})_{j \in \J, e \in E} \in L^2([0, H])^{\J \times E}$. Recall that this is a Hilbert space with scalar product
\[\scalar{x}{y} \coloneqq \sum_{j \in \J, e \in E} \int_0^H x_{j, e}(\xi) \cdot y_{j, e}(\xi) \diff \xi.\]

\paragraph{Infinite-dimensional variational inequality and weak-strong continuous mappings.} We want to briefly recall the main definitions in order to apply Br\'ezis' theorem:
For a
set of function vectors $\X \subseteq L^2([0, H])^{\J \times E}$ and a mapping $\A \colon \X \to L^2([a, b])^{\J \times
E}$ the variational inequality is the following task.
\begin{equation} \label{eq:VI:multi_commodity} \tag{VI}
\text{ Find } x' \in \X \text{ such that }\quad \scalar{\A(x')}{z' - x'} \geq 0 \quad \text{ for all } z' \in \X.\end{equation}
Br\'ezis \cite[Theorem 24]{brezis1968} specifies conditions to guarantee the existence of a solution (see also
\cite{tang2017}).
\begin{theorem} \label{thm:brezis:prelim}
Let $\X$ be a non-empty, closed, convex and bounded subset of $L^2(D)^d$. Let $\A : \X \rightarrow L^2(D)^d$ be a
weak-strong continuous mapping. Then, the variational inequality $\VI(\X,\A)$ has a solution $x^* \in \X$.
\end{theorem}

Hereby, we say a sequence $x'_k \in L^2(D)^d$ \defemph{converges weakly} to $x' \in L^2(D)^d$, if $\scalar{x'_k}{z'} \to
\scalar{x'}{z'}$ for all $z' \in L^2(D)^d$ (convergence in the weak-topology). For a given subset $\X \subseteq L^2(D)^d$, we call a mapping $\A\colon \X
\to L^2(D)^d$ \defemph{weak-strong continuous} at $x' \in \X$, if for every $x'_k \in \X$ that converges weakly to $x'$, we
have that $\A(x'_k)$ converges to $\A(x')$ with respect to the induced $L^2(D)^d$-norm. 
Note that the integration operator that maps $x'$ to the function $x$ defined by $x(\phi) \coloneqq \int_0^\phi x'(\xi) \diff \xi$ is a prominent example for a weak-strong continuous mapping.

\paragraph{Thin flow variational inequality.}
In order to obtain a variational inequality that solves the thin flow conditions, we start by defining
\[\X \coloneqq \Set{ (x'_{j, e})_{j \in \J, e \in E} \in L^2([0, H])^{\J \times E} | 
\begin{array}{r}(x'_{j, e}(\phi))_{e \in E} \text{ is a static $\source_j$-$\sink_j$-flow of}\\ [0.5em]
\text{value $1$ for $\phi \in \K_j$ and $0$ for $\phi \notin \K_j$.}\end{array}}.\]
Clearly, $\X$ is a non-empty, closed, convex and bounded subset of $L^2([0, H])^{\J \times E}$ since a convex combination
of two static flows with the same value is again a static flow of this value.

In order to define the mapping $\A$ we first need the following lemma.
\begin{restatable}{lemma}{lemximplieslmulticommodity} \label{lem:x'_implies_l:multi_commdity}
For every $x' = (x'_{j, e})_{j \in \J, e \in E} \in \X$ we can construct a vector $(\l_{j, v})_{j \in \J, v \in V}$ of
continuous and monotonically increasing functions such that their derivatives $(\l'_{j, v})_{j \in \J, v \in V}$ satisfy
\eqref{eq:l'_s:multi_commodity} and \eqref{eq:l'_v_min:multi_commodity} for all $\phi \in [0, H]$, where we define
\[\phi^i_{j, u}(\phi) \coloneqq \min \Set{ \varphi \geq 0 | \l_{i, u}(\varphi) = \l_{j, u}(\phi)} \quad \text{ and } \quad
y'_{j, e}(\phi) \coloneqq \sum_{i \in \J \setminus \set{j}} x'_{i, e} (\phi^i_{j, u} (\phi)) \cdot 
\frac{\l'_{j, u}(\phi)}{\l'_{i, u}(\phi^i_{j, u}(\phi))}.\] 
Furthermore, the mapping $(x'_{j, e})_{j \in \J, e \in E} \mapsto (\l_{j, u})_{j \in \J, u \in V}$ is weak-strong
continuous.
\end{restatable}

\proof{Proof of \Cref{lem:x'_implies_l:multi_commdity}.} \label{proof:x'_implies_l:multi_commodity}
The key idea of the proof is to start at time $0$ and then extend these functions step by step for later points in time.
Note that we extend over the range of the functions and not over the domain. In each extension step we determine the
change of $\l_{j, v}$ by plugging $x'_{j, e}(\l_u(\phi))$ and $y'_{j, e}(\l_u(\phi))$ into
\eqref{eq:l'_v_min:multi_commodity}. As we assumed that the transit time of every arc is positive we have that $\l_{j,
u}(\phi) < \l_{j, v}(\phi) = \theta$ for all arcs that are active for $j$ and $\phi$. Hence, we can extend these
functions at least by the minimal transit time in every step.

For the formal proof, we initialize $\l_{j, v}(0)$ with the shortest
distance from $\source_j$ to $v$, only considering the transit times. For technical reasons we define the
$\l$-functions also for negative values by setting $\l_{j, u}(\phi) \coloneqq \l_{j, u}(0) - \frac{\phi}{\inrate_j}$ for
all $\phi < 0$. Furthermore, we assume that $x'$ is defined on $\R$ with $x'_{j, e}(\phi) = 0$ for all $\phi \notin [0,
H]$.

For the extension step suppose that there is a $\theta_0 \geq 0$ such that each earliest arrival time
function~$\l_{j, v}$ is already defined on an interval $(- \infty, \phi_{j, v}]$ with $\l_{j, v}(\phi_{j, v}) = \theta_0$ and
that these $\l$-functions satisfy the condition of the lemma on this interval. Clearly, this is given for $\theta_0 = 0$ as
\eqref{eq:l'_s:multi_commodity} and \eqref{eq:l'_v_min:multi_commodity} only have to hold for non-negative $\phi$.

We are going to extend each function $\l_{j, v}$ such that the properties hold up to some $\theta_0 + \alpha$.

For $v = \source_j$ we set
\[\l_{j, \source_j}(\phi) \coloneqq \frac{\phi}{\inrate_j}.\]
In order to extend $\l_{j, v}$ for $v \neq \source_j$ we consider all incoming arcs $e = uv \in \delta_v^-$ that are
active for $\phi_{j, u}$ according to the function values from the past that are defined already. In other words, we
define
\[\delta'_v \coloneqq \Set{e=uv \in \delta_v^- | \l_{j, u}(\phi_{j, u}) + \tau_e \leq \l_{j, v}(\phi_{j, u})}.\]

Since we consider strictly positive transit times we have $\l_{j, v}(\phi_{j, u}) \geq \l_{j, u}(\phi_{j, u}) +
\tau_e = \theta_0 + \tau_e > \theta_0$, and hence, $\phi_{j, v} < \phi_{j, u}$, for all $e = uv \in \delta'_v$.

We define for all $\phi \in [\phi_{j, v}, \phi_{j, u})$
\[\phi^i_{j, u}(\phi) \coloneqq \min \Set{ \varphi \geq 0 | \l_{i, u}(\varphi) = \l_{j, u}(\phi)} \quad \text{ and } \quad
y'_{j, e}(\phi) \coloneqq \sum_{i \in \J \setminus \set{j}} x'_{i, e} (\phi^i_{j, u} (\phi)) 
\cdot \frac{\l'_{j, u}(\phi)}{\l'_{i, u}(\phi^i_{j, u}(\phi))}.\] 

Here, $\l'_{j, u}$ and $\l'_{i, u}$ are the derivatives of the corresponding functions $\l_{j, u}$ and $\l_{i, u}$,
which are well-defined on $(- \infty, \phi_{j, u})$ and $(- \infty, \phi^i_{j, u}(\phi_{j, u}) )$, respectively, as
$\l_{i, u} (\phi^i_{j, u}(\phi_{j, u})) = \l_{j, u}(\phi_{j, u}) = \theta_0$.

To determine the earliest arrival time of $\phi$ at $v$ when using arc $e = uv$ we define
\[\rho_{j, e}(\phi) \coloneqq \begin{cases}
 \frac{x'_{j, e}(\phi) + y'_{j, e}(\phi)}{\capa_e} & \text{ if } \l_{j, u}(\phi) + \tau_e < \l_{j, v}(\phi),\\
\max\Set{\l'_{j, u}(\phi), \frac{x'_{j, e}(\phi) + y'_{j, e}(\phi)}{\capa_e}} & \text{ else. }\\
\end{cases}\]

Finally, we extend the earliest arrival time $\l_{j, v}$ for $\phi \in (\phi_{j, v}, \phi_{j, v} + \epsilon]$ by
\[\l_{j, v}(\phi) \coloneqq \l_{j, v}(\phi_{j, v}) 
+ \min_{e \in \delta'_v} \int_{\phi_{j, v}}^{\phi} \rho_{j, e}(\xi) \diff \xi.\]
If we choose $\epsilon$ to be small enough, such that $\phi_{j, v} + \epsilon \leq \phi_{j, u}$ for all $u$ with $uv \in
\delta'_v$, the right side is always well-defined. Clearly, the extended function $\l_{j, v}$ is continuous and
monotonically increasing, and by construction it satisfies \eqref{eq:l'_v_min:multi_commodity}, since an active arc has
a positive waiting time at $\l_{j, u}(\phi)$ if, and only if, $\l_{j, u}(\phi) + \tau_e < \l_{j, v}(\phi)$.

Note that all $\l'_{j, v}$ are bounded from above, as $x'_{j,e}$ is bounded, and therefore, there exists an $\alpha >
0$ independent of $\theta_0$ such that we can extend all $\l$-functions to $\theta_0 + \alpha$.
By iteratively applying this extension step we end up with $\l$-functions that are at least defined on $[0, H]$.

As the procedure only depends on the $x$-functions this construction provides a mapping $x' \mapsto \l$, which is
weak-strong continuous as the integration operator on compact intervals is weak-strong continuous in $L^2$.
Furthermore, all operations we used, such as taking sums, minima, maxima and doing time-shifting are continuous mappings
when considering the $L^2$-norm. As $x' \to \l$ is a concatenation of continuous functions with a weak-strong continuous
function it is also weak-strong continuous.
\hfill\Halmos\endproof
 
\Cref{lem:x'_implies_l:multi_commdity} shows that for a given $x' \in \X$ we obtain functions $\l_{j, v}$, which we can
plug into \Cref{lem:q_characterized_by_l:multi_commodity} in order to obtain waiting time functions $(q_e)_{e \in E}$.
Note that the mapping $x' \mapsto q$ is also weak-strong continuous. Furthermore, the $\l$-functions satisfy
\Cref{eq:bellman:multi_commodity}, as we have already shown in the second half of the proof of
\Cref{thm:thin_flows_form_Nash_flows:multi_commodity} (we do not use \eqref{eq:l'_v_tight:multi_commodity} in this part
of the proof). It is worth noting, however, that these $\l$- and $q$-functions do not belong to a feasible flow over
time, in general, as flow conservation might not hold when deriving in- and outflow rate functions in the usual way.

Finally, we can define the weak-strong continuous mapping $\A \colon \X \to L^2([0, H])^{\J \times E}$ by
\[(x'_{j, e})_{j \in \J, e \in E} \mapsto (h_{j, e})_{j \in \J, e \in E} \quad \text{ with } \quad h_{j, e}(\phi) 
\coloneqq  \l_{j, u}(\phi) + \tau_e + q_e(\l_{j, u}(\phi)) - \l_{j, v}(\phi).\]

In other words, if $x'$ corresponds to a feasible flow over time, $h_{j, e}(\phi)$ denotes the delay of particle~$\phi$
when traveling as fast as possible to $u$ first and then using arc $e$, instead of taking the fastest direct route to
$v$. In a Nash flow over time this value should always be $0$ for each arc $e$ with $x'_{j, e}(\phi) > 0$.

\begin{theorem}\label{thm:existence_of_multi_commodity_thin_flow:multi_commodity}
For every multi-commodity network with positive transit times there exists a multi-commodity thin flow, and hence, a
multi-commodity Nash flow over time.
\end{theorem}
\proof{Proof of \Cref{thm:existence_of_multi_commodity_thin_flow:multi_commodity}.}
Let $x'$ be a solution to the variational inequality constructed above, which exists due to \Cref{thm:brezis:prelim}. In
other words, it holds that
\[\sum_{e \in E, j \in \J} \int_0^H (\l_{j, u}(\phi) + \tau_e + q_e(\l_{j, u}(\phi)) 
- \l_{j, v}(\phi) ) \cdot (z'_{j, e} - x'_{j, e}) \diff \xi \geq 0 \qquad \text{ for } z' \in \X.\]
Let $(\l_v)_{v \in V}$ be the node labels corresponding to $x'$ according to \Cref{lem:x'_implies_l:multi_commdity} with
derivatives $(\l'_v)_{v \in V}$. We will show that $(x', \l')$ satisfies the multi-commodity thin flow conditions for
$\phi \in [0, H]$.

As \eqref{eq:l'_s:multi_commodity} and \eqref{eq:l'_v_min:multi_commodity} hold for $(\l'_v)_{v \in V}$ by
\Cref{lem:x'_implies_l:multi_commdity} it only remains to show that \eqref{eq:l'_v_tight:multi_commodity} holds for
almost all $\phi \in [0, H]$. In order to do so, suppose that there exist a commodity $j$, an arc $e = uv$ and a set
with positive measure $\Phi \subseteq [0, H]$ such that $x'_{j, e}(\phi) > 0$ and
\[\l'_{j,v}(\phi) < \rho_{j, e}^\phi\left(\l'_{j,u}(\phi), x'_{j, e}(\phi), y'_{j, e}(\phi)\right).\]
We assume that $\Phi$ is contained in a small interval $[a, b]$ and that $x'_{j, e}(\phi) \geq \epsilon$ for some
$\epsilon > 0$.

Note that for every $\phi \in \Phi$ there are two $\source_j$-$\sink_j$-paths $P_\phi, Q_\phi$, which satisfy the
following conditions. Firstly, we require that $e \in P_\phi$ and $x'_{j, e'}(\phi) > \epsilon$ for all $e' \in P_\phi$,
and secondly, for all $e' = u'v' \in Q_\phi$ we demand that $e' \in E'_{j, \phi}$ as well as
\[\l'_{j,v'}(\phi) = \rho_{j, e'}^\phi\left(\l'_{j,u'}(\phi), x'_{j, e'}(\phi), y'_{j, e'}(\phi)\right).\]
The existence of $P_\phi$ follows by the flow conservation of the static flow $x'_{j, e}(\phi)$ ($\epsilon$ can be
redefined to be small enough) and the existence of $Q_\phi$ follows by the construction of the $\l'$-functions.

It is possible to partition $\Phi$ into measurable sets such that the particles $\phi$ of each subset have the same
paths-pair $(P_\phi, Q_\phi)$. Thus, at least one of these subsets has to have a positive measure, and hence, without
loss of generality, we can assume that all particles in $\Phi$ have the same pair of paths, which we denote by $P$ and
$Q$.

We set $z' \coloneqq x'$ with the exception of $z'_{j, e'}(\phi) \coloneqq x'_{j, e'}(\phi) - \epsilon$ for all $e' \in
P$. Furthermore, let $\phi \in \Phi$ and $z'_{j, e'}(\phi) \coloneqq x'_{j, e'}(\phi) + \epsilon$ for all $e' \in Q$ and
$\phi \in \Phi$. Clearly, $z' \in K$, as the small shift of flow from $P$ to $Q$, does not violate the flow conservation
and does not change the total flow value. We obtain that
\begin{align*}
\scalar{\A (x')}{z' - x'} &=\sum_{e \in E, j \in \J} \int_0^H T_e(\l_{j, u}(\phi)) \cdot (z'_{j, e} - x'_{j, e}) \diff \xi \\
&= - \epsilon \cdot \sum_{e' = u'v' \in P} \int_{\Phi} T_{e'}(\l_{j, u'}(\phi)) - \l_{j, v'}(\phi) \diff \phi 
+ \epsilon \cdot \sum_{e' \in Q} \int_{\Phi} T_{e'}(\l_{j, u'}(\phi))- \l_{j, v'}(\phi) \diff \phi\\
&\leq - \epsilon \cdot \int_{\Phi} T_e(\l_{j, u}(\phi))- \l_{j, v}(\phi)\diff \phi 
+ \epsilon \cdot \sum_{e' \in Q} \int_{\Phi} T_{e'}(\l_{j, u'}(\phi))- \l_{j, v'}(\phi) \diff \phi\\
&= - \epsilon \cdot \int_{\Phi} T_e(\l_{j, u}(\phi))- \l_{j, v}(\phi) < 0.
\end{align*}

The first inequality follows, since $\l$ satisfies \eqref{eq:bellman:multi_commodity}, and hence, $T_{e'}(\l_{j,
u'}(\phi))- \l_{j, v'} \geq 0$ for all $e' = u'v' \in E$. The last equation holds, since $Q$ is a path of active arcs
for all particles in $\Phi$, and therefore, $T_{e'}(\l_{j, u'}(\phi))- \l_{j, v'}(\phi) = 0$ for all $e' = u'v' \in Q$
and all $\phi \in \Phi$. But this is a contradiction to the variational inequality \eqref{eq:VI:multi_commodity}. Hence,
$(x', \l')$ satisfies the thin flow conditions for almost all $\phi \in [0, H]$.

This shows the existence of a multi-commodity thin flow and with \Cref{thm:thin_flows_form_Nash_flows:multi_commodity}
it follows that there also exists a multi-commodity Nash flow over time on every multi-commodity network.
\hfill\Halmos\endproof

\new{\begin{remark}\label{remark:piecewise_constant_thin_flows}
It is worth noting that this existence proof only works for bounded inflows durations $I_j$ as the set $\X$ must be a
bounded set. Moreover, it is not possible (at least not in a straight-forward manner) to obtain a Nash flow over time
for unending inflow rates by constructing a sequence of Nash flows over time with increasing inflow durations $I_j$. The
reason for this is the following. Since flow entering at a later point in time can significantly influence the route of
flow entering earlier, the Nash flows over time in such a sequence would not necessarily be
nested. But this is an essential prerequisite to apply a transfinite induction as it was used for the single-commodity
case; see \cite{cominetti2011existence}. Hence, the existence of Nash flows over time with unending inflow remains an open problem.
\end{remark}}

\paragraph{Structure of multi-commodity Nash flows over time.} Unfortunately, the properties of the $x'$- and $\l'$-functions of a multi-commodity thin flow over time remain unknown.
We conjecture that the earliest arrival times $\l$ must be piece-wise linear, or in other words, the $\l'$
functions must be piece-wise constant. The intuition behind this is the following consideration. Whenever a current
shortest path changes for some commodity, this means that either an arc became active or a queue depleted. It seems
that every arc can only be responsible for a countable amount of events and every event will cause at most a countable
amount of jump points for the $\l'$-functions. Unfortunately, we were not able to prove this yet.

\newpage
\section{Connection to common destination or common origin settings.} \label{sec:common_destination_or_origin} Even though multi-commodity Nash flows over time exist, we do not
know how to construct them, as exact solutions to infinite-dimensional variational inequalities cannot be computed
algorithmically. In order to transfer the constructive existence proof of Koch-Skutella-model to a multi-terminal setting, we considered, in collaboration with Skutella, the
special case where all commodities have the same destination or the same origin; see \cite{sering2018multiterminal}. 
\new{In this section we want to briefly recall these special cases and show that they, indeed, provide a multi-commodity Nash flow over time as defined in \Cref{def:Nash_flow:multi_commodity}.}

\subsection{Common destination.} \label{subsec:multi_source}
In the common-destination-setting we have multiple sources but only one
sink. We consider only one commodity, i.e., only one flow, where each particle can choose the source for entering the
network; see \Cref{fig:setting:multi_source}. Sources further away from the sink might not be chosen by particles with
high priority, but their network inflow rates will be maximal from time $0$ onwards nonetheless.

\begin{figure}[b]
\centering \includegraphics{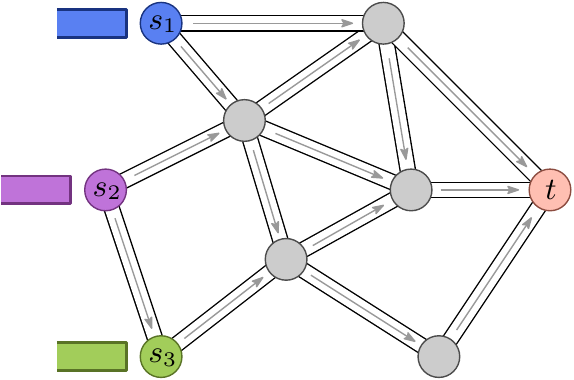} \qquad \qquad \includegraphics{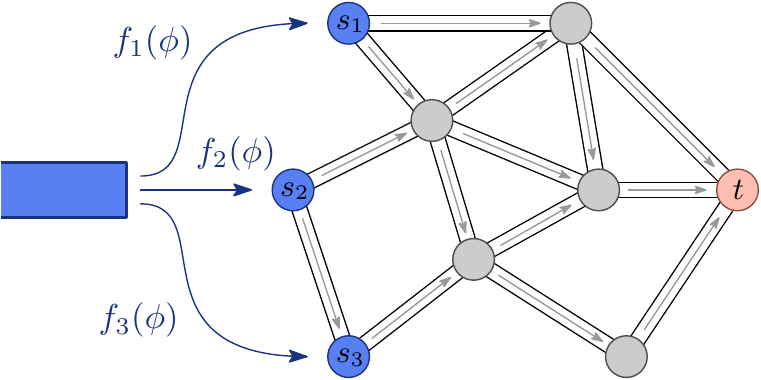}
\caption{\emph{On the left:} A multi-commodity network where each commodity has the same destination $\sink$. \emph{On
the right:} In this case we can construct Nash flows over time by reducing the problem to a single commodity problem.
The inflow distribution denotes the proportion of the flow that enters through each source.}
\label{fig:setting:multi_source}
\end{figure}

\paragraph{Evacuation scenarios.} Even though this setting seems to be artificial at first glance it has an important
and relevant application. Imagine an inhabited region with a high risk of flooding, where in the case of rising water
levels everyone tries to reach a high-altitude shelter as fast as possible. Since the road users seek for protection and
do not care at which of the shelters they end up, we can connect all nodes representing one of these safe places to a
super sink. Without regulations an evacuation now corresponds to a Nash flow over time with multiple sources but only a
single sink, as everyone starts at their home and tries to reach one of the shelters as fast as possible.

\paragraph{Networks.} To model this special case we consider a network $G = (V, E)$ with transit times $\tau_e \geq 0$,
capacities $\capa_e > 0$ and a sink node $\sink$ as before. But this time we have, in addition, a set of sources
$\Sources \coloneqq \Set{\source_j | j \in J}$ with network inflow rates $\inrate_j$ for each commodity $j$. We assume
that every source can reach the sink, that every node is reachable by at least one source and that all directed cycles
have a strictly positive total transit time.

We do not distinguish between different commodities because as soon as the particles have entered the network they all
have the same goal, namely to reach the sink as fast as possible, and therefore, their identity is interchangeable.

We use the same notation of flow rates, cumulative flows, queue sizes and waiting times as before, this time,
however, we say that flow is conserved on a node $v \in V \setminus \set{\sink}$ if
\begin{equation} \label{eq:flow_conservation:multi_source}
\sum_{e\in \delta^+_v} f_e^+(\theta) - \sum_{e \in \delta^-_v} f_e^-(\theta) = \begin{cases}
0 & \text{ if } v \in V\setminus \Sources, \\
\inrate_j & \text{ if } v = \source_j \in \Sources.
\end{cases}\end{equation}

A \defemph{flow over time} is a family of locally integrable and bounded functions $(f_e^+, f_e^-)_{e \in E}$ that
satisfies \eqref{eq:conservation_on_arc:multi_commodity} (for a single commodity only) and \eqref{eq:flow_conservation:multi_source} and it is
\defemph{feasible} if \eqref{eq:totaloutflow:multi_commodity} is fulfilled.

\paragraph{Inflow distributions.} In order to denote which particle enters through which source, we introduced a new set of functions in \cite{sering2018multiterminal}. A family of locally integrable functions $f_j\colon \Flow \to [0, 1]$, for $j \in \J$, is called
\defemph{inflow distribution} if $\sum_{j \in \J} f_j(\phi) = 1$ for almost all $\phi \in \Flow$ and if each
\defemph{cumulative source inflow} $F_j(\phi) \coloneqq \int_{0}^{\phi} f_j(\varphi) \diff \varphi$ is unbounded
for~$\phi \to \infty$. The function $f_j(\phi)$ describes the fraction of particle~$\phi$ that enters the network trough
$\source_j$. The cumulative source inflow functions have to be unbounded in order to guarantee that the inflow rates at
the sources never run dry.

\paragraph{Source arrival times.} Given a feasible flow over time $f$, the \defemph{source arrival time} functions map
each particle $\phi \in \Flow$ to the time it arrives at $\source_j$ and they are given by
\[\T_j(\phi) \coloneqq \frac{F_j(\phi)}{\inrate_j}.\]  

\paragraph{Earliest arrival times.} 
The \defemph{earliest arrival times} are now defined by
  \begin{align}  \label{eq:bellman:multi_source}
  \begin{aligned}
  \l_{\source_j}(\phi) &= \;\min\left(\{\:\T_j(\phi)\:\} \cup \Set{\T_e(\l_u(\phi)) | e = u \source_j\in E}\right)
  &&\quad \text{ for } j \in \J, \\
  \l_v(\phi) &= \!\!\min_{e = u v\in E} \T_e(\l_u(\phi))  &&\quad \text{ for } v
  \in V\backslash \Sources.
  \end{aligned}
  \end{align}
  This is well-defined since all cycles in $G$ have positive travel times by assumption.

\paragraph{Current shortest paths networks and active arcs.} As before we call an arc $e = uv$ \defemph{active}
for~$\phi$ if $\l_v(\phi) = \T_e(\l_u(\phi))$ holds and we denote the set of all active arcs for a particle~$\phi$ by
$E'_\phi$ as well as the \defemph{current shortest paths network} by $G'_\phi = (V, E'_\phi)$. Furthermore, $E^*_\phi
\coloneqq \!\set{\!e = uv \in E| q_e(\l_u(\theta)) > 0\!}$ denotes the set of \defemph{resetting arcs}.

\paragraph{Multi-source Nash flows over time.} A dynamic equilibrium now consists of a feasible flow over
time together with an inflow distribution and each particle chooses a convex combination of routes from the sources to
the sink such that it arrives there as fast as possible.

\begin{definition}[Multi-source Nash flow over time] \label{def:Nash_flow:multi_source} 
A tuple $f = ((f^+_e)_{e\in E},(f_j)_{j \in \J})$ consisting of a feasible flow over time and an inflow distribution  is a
\defemph{multi-source Nash flow over time} if the following two \defemph{Nash flow conditions} hold:
\begin{align}
\l_{\source_j}(\phi) &= \T_j(\phi)  &&\text{ for all } j \in \J \text{ and almost all } \phi \in \Flow,
\label{eq:Nash_condition_source:multi_source} \tag{msN1}\\
f^+_e(\theta) &> 0 \;\;\Rightarrow\;\; \theta \in \l_u(\Phi_e) &&\text{ for all arcs } e = uv \in E \text{ and almost all }
\theta \in [0, \infty),\tag{msN2}\label{eq:Nash_condition_active:multi_source} 
\end{align}
where $\Phi_e \coloneqq \set{\phi \in \Flow | e \in E'_\phi}$ is the set of flow particles for which arc $e$ is
active.
\end{definition}

Figuratively speaking, these two conditions mean that entering the network through a source $\source_j$ is always a
fastest way to reach $\source_j$ \eqref{eq:Nash_condition_source:multi_source} and that a Nash flow over time uses only
active arcs \eqref{eq:Nash_condition_active:multi_source}, and therefore only shortest paths, to $\sink$.

\begin{lemma}[cf. Lemma 2 in \cite{sering2018multiterminal}] \label{lem:nash_flow_characterization:multi_source}
A tuple~$f = ((f^+_e)_{e\in E}, (f_j)_{j \in \J})$ of a feasible flow over time and an inflow distribution is a Nash
flow over time if, and only if, we have
\[F_e^+(\l_u(\phi)) = F_e^-(\l_v(\phi)) \quad \text{ and } \quad F_j(\phi) = \l_{\source_j}(\phi) \cdot \inrate_j\]
for all arcs $e = uv \in E$, every $j \in \J$, and all particles $\phi \in \Flow$.
\end{lemma}

\paragraph{Underlying static flows.} 
The underlying static flow is now given by two types of functions
\[x_e(\phi) \coloneqq F_e^+(\l_u(\phi)) =  F_e^-(\l_v(\phi)) \quad \text{ and } 
\quad x_j(\phi) \coloneqq F_j(\phi) = \l_{\source_j} (\phi) \cdot \inrate_j.\]
For every $\phi$ this is a static $\Sources$-$\sink$-flow with $x_j(\phi)$ as supply at source $\source_j$ since the
integral of \eqref{eq:flow_conservation:multi_source} over $[0, \l_v(\phi)]$ yields
\begin{equation} \label{eq:x_is_flow:multi_source}
\sum_{e\in \delta^+_v} x_e(\phi) - \sum_{e \in \delta^-_v} x_e(\phi) = \begin{cases}
0 & \text{ if } v \in V \setminus (\Sources \cup \set{\sink}), \\ 
\l_{\source_j} (\phi) \cdot r_j = x_j(\phi) & \text{ if } v = \source_j \in \Sources.
\end{cases}
\end{equation}

Let $x'_e$, $x'_j$ and $\l'_v$ denote the derivative functions, then, it is possible to determine the inflow function of every arc $e = uv$ as well as the
inflow distribution from these derivatives, since
\[x'_e(\phi) = f_e^+(\l_u(\phi)) \cdot \l'_u(\phi) \quad \text{ and } \quad f_j(\phi) = 
\l'_{\source_j}(\phi) \cdot \inrate_j.\]
Consequently, a Nash flow over time is, again, completely characterized by these derivatives.
Differentiating~\eqref{eq:x_is_flow:multi_source} yields that $x'(\phi)$ also forms a static $\Sources$-$\sink$-flow,
which we consider next.

\paragraph{Multi-source thin flows with resetting.} Let $E' \subseteq E$ be a subset of arcs
such that the subgraph $G' = (V, E')$ is acyclic and every node is reachable by some source within $G'$. Note that not
every node needs to be able to reach sink~$\sink$. Additionally, we consider a subset of resetting arcs $E^* \subseteq
E'$. Moreover, let $\X\left(E', (x'_j)_{j \in \J}\right)$ be the set of all static $\Sources$-$\sink$-flows in~$G'$ with
supply~$x'_j$ at source~$\source_j$ for~$x'_j \geq 0$ and $\sum_{j \in \J} x'_j = 1$.

\begin{definition}[Multi-source thin flow with resetting] \label{def:thin_flow:multi_source}
A vector $(x'_j)_{j \in \J}$ with $x'_j \geq 0$ and $\sum_{j \in \J} x'_j = 1$, together with a static flow $(x'_e)_{e \in E}
\in K\left(E', (x'_j)_{j \in \J}\right)$ and a node labeling $(\l'_v)_{v \in V}$ is called \defemph{multi-source thin flow with
resetting} on $E^* \subseteq E'$ if
\begin{alignat}{2}
\l'_{\source_j} &= \frac{x'_j}{\inrate_j} &&\text{ for all } j \in \J, \label{eq:l'_s:multi_source} \tag{msTF1}\\ 
\l'_{\source_j} &\leq \min_{e = u\source_j \in E'} \rho_e(\l'_u, x'_e) \quad && \text{ for all } j \in \J,
\label{eq:l'_s_min:multi_source}\tag{msTF2}\\ 
\l'_v &= \min_{e = uv \in E'} \rho_e(\l'_u, x'_e) \quad && \text{ for all } v \in V \backslash \Sources,
\label{eq:l'_v_min:multi_source}\tag{msTF3}\\ 
\l'_v &= \rho_e(\l'_u, x'_e) && \text{ for all } e = uv \in E' \text{ with } x'_e > 0,
\label{eq:l'_v_tight:multi_source}\tag{msTF4}
\end{alignat}
\[\text{ where } \qquad \rho_e(\l'_u, x'_e) \coloneqq \begin{cases}
\frac{x'_e}{\capa_e} & \text{ if } e = uv \in E^*,\\
\max\Set{\l'_u, \frac{x'_e}{\capa_e}} & \text{ if } e = uv \in E'\backslash E^*.  
\end{cases}\]
\end{definition}

\new{We proved in \cite[Theorem 4]{sering2018multiterminal} that the derivatives of a multi-source Nash flow over time $f$ form a multi-source thin flow with resetting almost everywhere. Furthermore, it is possible to construct a multi-source Nash flow over time by extending it step by step with the help of multi-source thin flows with resetting \cite[Theorem 5, 7, 8]{sering2018multiterminal}.}

\paragraph{Multi-commodity Nash flows over time with common destination.} As the main contribution in this multi-source-setting we want to show that these
multi-source Nash flows over time do indeed correspond to a multi-commodity Nash flow over time where all
commodities share the same destination. To do so, consider a multi-source Nash flow over time $f = (f_e^+,
f_e^-)$ as constructed above with a multi-source thin flow $(x'(\phi), \l'(\phi))$ for each particle $\phi$. By adding a super source
$\source$ and a new arc $e_j = \source \source_j$ carrying a flow of $x'_j(\phi)$ for each $j \in \J$, we obtain an
$\source$-$\sink$-flow of value $1$ and by using the flow decomposition theorem we obtain a path-based
formulation~$(x'_P)_{P \in \Paths}$. For every $j \in \J$ let $\Paths_j$ be all $\source$-$\sink$-paths that start with
the new arc $e_j$. By assigning all flow on these paths to commodity~$j$ we obtain
\[x'_{j, e} \coloneqq \sum_{\substack{P \in \Paths_j\\\text{with } e \in P}} x_P.\]
Setting 
\[f^+_{j, e}(\theta) \coloneqq \frac{x'_{j. e}(\phi)}{\l'_u(\phi)} \quad \text{for } \theta 
= \l_u(\phi) \qquad \text{ and } \qquad
f^-_{j, e}(\theta) \coloneqq \frac{x'_{j, e}(\phi)}{\l'_v(\phi)} \quad \text{for } \theta = \l_v(\phi)\]
for all $\phi \in \Flow$ provides a multi-commodity Nash flow over time with unlimited inflow rates as we show in the
following theorem.

\begin{theorem}\label{thm:multi_source_is_multi_commodity_nash_flow:multi_source}
The family of functions $(f_{j, e}^+, f_{j, e}^-)_{j \in \J, e \in E}$ is a multi-commodity Nash flow over time in a
network with inflow rates~$\inrate_j$ and $I_j = [0, \infty)$ for all $j \in \J$.
\end{theorem}
\proof{Proof of \Cref{thm:multi_source_is_multi_commodity_nash_flow:multi_source}.}
Flow conservation \eqref{eq:flow_conservation:multi_commodity} for every commodity $j \in J$ on every node $v \in V
\setminus \set{\sink}$ follows immediately since
\[\sum_{e\in \delta^+_v} f_{j, e}^+(\theta) - \sum_{e \in \delta^-_v} f_{j, e}^-(\theta) 
= \sum_{e\in \delta^+_v} \frac{x'_{j, e}}{\l'_{j, v}} - \sum_{e \in \delta^-_v} \frac{x'_{j, e}}{\l'_{j, v}} 
= \begin{cases}
0 & \text{ if } v \neq \source_i,\\
\frac{x'_j}{\l'_v} = \inrate_j & \text{ if } v = \source_j.
\end{cases}\]
  
The condition on the total outflow rate \eqref{eq:totaloutflow:multi_commodity} follows immediately, since the total flow is a feasible single-commodity flow over time. Condition~\eqref{eq:FIFO:multi_commodity} is satisfied for each arc $e = uv \in E$
since in the case of $f_e^+(\l_u(\phi)) > 0$ we have that $e$ is active (from the single-commodity perspective), i.e.,
$\l_v(\phi) = T_e(l_u(\phi))$, and therefore,
\[f^-_{j, e}(\l_v(\phi)) = \frac{x'_{j, e}(\phi)}{\l'_v(\phi)} 
= \frac{x'_e(\phi)}{\l'_v(\phi)} \cdot \frac{x'_{j, e}(\phi)}{\l'_u(\phi)} \cdot \frac{\l'_u(\phi)}{x'_e(\phi)} 
= f^-_e(\l_v(\phi)) \cdot \frac{f^+_{j, e}(\l_u(\phi))}{f_e^+(\l_u(\phi))}.\]
For $f_e^+(\l_u(\phi)) = 0$ we clearly have $x'_{i, e} = 0$, and thus, $f^-_{i, e}(\l_v(\phi)) = 0$. Note that the
function $\l_v$ is continuous and unbounded, and therefore, we can found for every $\theta \geq \l_v(0)$ a $\phi \in
\Flow$ with $\theta = \l_v(\phi)$. For $\theta < \l_v(0)$ we have $f_{j, e}^+(\theta) = f_{j, e}^-(\theta) = 0$ as no
flow has reached $e$ yet. Hence, we have a feasible multi-commodity flow over time.

The multi-commodity Nash flow condition \eqref{eq:Nash_condition:multi_commodity} follows immediately by
\Cref{lem:nash_flow_characterization:multi_commodity} and by
\[ F^+_{j, e}(\l_u(\phi)) = \int_0^{\phi} f_{j, e}^+(\l_u(\xi)) \cdot \l'_u(\xi) \diff \xi 
= \int_0^\phi x'_{j, e}(\xi) \diff \xi = \int_0^{\phi} f_{j, e}^-(\l_v(\xi)) \cdot \l'_v(\xi) \diff \xi 
= F^-_{j, e}(\l_v(\phi)),\]
which holds for all $\phi \in \Flow$.
\hfill\Halmos\endproof

\subsection{Common origin.} \label{subsec:multi_sink} In this last subsection we are going to consider the second special
case analyzed in \cite{sering2018multiterminal}, namely multiple commodities with a common origin. Suppose that each
commodity has its own sink but all flow starts at the same common source; see left the side of
\Cref{fig:setting:multi_sink}. In this scenario the commodities matter a lot, since different flow particles within the
network might want to reach different sinks. Nonetheless, it was shown that this special case, once again, can be
reduced to the single-commodity case by using a super sink construction as it is shown on the right side of
\Cref{fig:setting:multi_sink}. \new{The contribution of this paper is to simplify the notation and, as a new addition,
we show that the different commodities can be again reconstructed by using path decompositions of the thin flows. This
shows, that we indeed obtain a multi-commodity Nash flow over time with common origin.}

\begin{figure}[t]
\centering \includegraphics{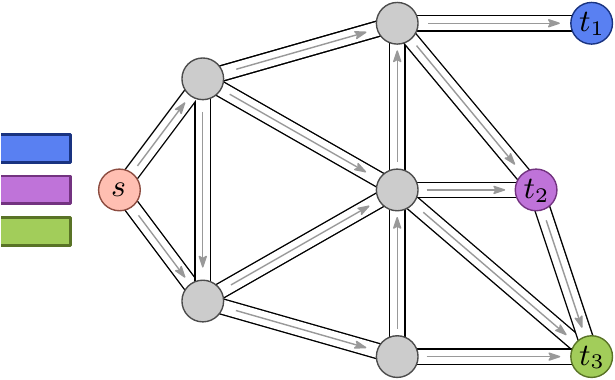} \qquad \includegraphics{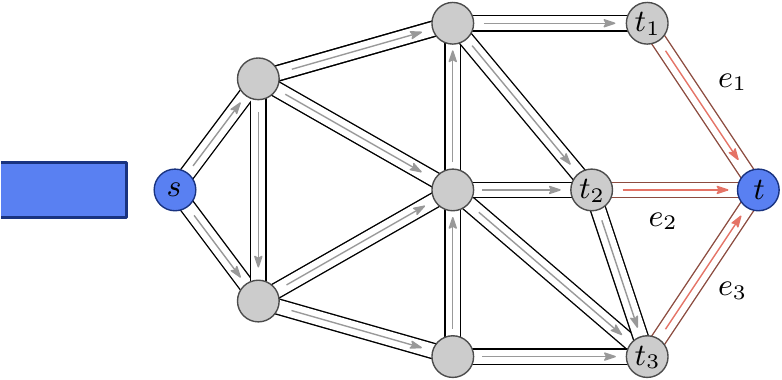}
\caption{\textit{On the left:} A multi-commodity network where all commodities share the same origin. \textit{On the
right:} Constructing a multi-commodity Nash flow over time in this setting can be reduced to a single commodity Nash
flow over time by adding a super sink $t$ and new arcs $e_j$ with very small capacities.} \label{fig:setting:multi_sink}
\end{figure}

\paragraph{Extended graphs.} The key idea for constructing a Nash flow over time in this setting is to add a super sink to the
graph. For this let~$\capa_{\min} \coloneqq \min_{e\in E} \capa_e$ be the minimal capacity of the network, $\inrate
\coloneqq \sum_{j \in \J} \inrate_j$ the total network inflow and $\sigma \coloneqq \min\set{\capa_{\min}, \inrate}$.
For all $j\in \J$ we define $\delta_j$ to be the length of a shortest $\source$-$\sink_j$-path according to the transit
times. Furthermore, let $\delta_{\max} \coloneqq \max_{j \in \J} \delta_j$ be the maximal distance from the source to a
sink. We extend $G$ by a super sink~$\sink$ and $\abs{\J}$ new arcs $e_j \coloneqq (\sink_j, \sink)$ with
\begin{equation}\label{eq:definition_new_arcs:mult_sink}
\tau_{e_j} \coloneqq \delta_{\max} - \delta_j \quad \text{ and } \quad \capa_{e_j} 
\coloneqq \frac{\inrate_j \cdot \sigma}{2\inrate}.
\end{equation}
The \defemph{extended graph} is denoted by $\exG \coloneqq (\exV, \exE)$ with $\exV \coloneqq V \cup \Set{\sink}$ and
$\exE \coloneqq E \cup \Set{e_1, \dots, e_\m}$.
    
Note that the new capacities are strictly smaller than all original capacities and that they
are proportional to the inflow rate of the respective commodity. Furthermore, the transit times are chosen in such a way that all
new arcs are in the current shortest paths network for particle $\phi = 0$. The reason for the choice of $\sigma$ is the
following.
\begin{lemma} \label{lem:bound_for_l':multi_sink}
For every single-commodity thin flow with resetting $(x',\l')$ (see \Cref{def:thin_flow:multi_source} but with $\l'_\source = \frac{1}{r}$ instead of \eqref{eq:l'_s:multi_source} and \eqref{eq:l'_s_min:multi_source}) in $\exG$ it
holds for all $v \in V \setminus \set{\sink}$ that $\l'_v \leq \frac{1}{\sigma}$.
\end{lemma}
\proof{Proof of \Cref{lem:bound_for_l':multi_sink}.}
We have $\l'_s = \frac{1}{\inrate} \leq \frac{1}{\sigma}$ and $\frac{x_e'}{\capa_e} \leq \frac{1}{\capa_{\min}} \leq
\frac{1}{\sigma}$. Hence, by induction over the acyclic current shortest paths network we obtain that $\l'_v \leq
\max\set{\frac{x_e'}{\capa_e}, \l'_u} \leq \frac{1}{\sigma}$ for all active arcs $e = uv$.
\hfill\Halmos\endproof

\paragraph{Reduction from single-commodity Nash flows over time.} We obtain a multi-commodity Nash flow over time $f$
with common source by using a single-commodity Nash flow over time $\exf$ in $\exG$, which exists due to
\cite{cominetti2011existence}. To prove this we first show that, if all new arcs are active for some particle
$\phi$, then there is a static flow decomposition of the single-commodity thin flow with resetting~$x'$ with $x'_{e_j} =
\inrate_j$. This is formalized in the following lemma, where we write $x'\big|_E$ for the restriction of $x'$ to the
original graph $G$ and $\abs{\,\cdot\,}$ for the flow value of a static flow.
\begin{restatable}[cf. Lemma 11 in \cite{sering2018multiterminal}]{lemma}{lemthinflowsdecompositionmultisink} \label{lem:thin_flow_decomp:multi_sink}
  For any $E^* \subseteq E' \subseteq \exE$ with $\set{e_j|j\in \J} \subseteq E'$ consider the single-commodity thin
  flow with resetting $(x', \l')$ with network inflow rate $\inrate$. There exists a static flow decomposition~$(x'_{j,
  e})_{j \in J, e \in E}$ with $x'\big|_E = \sum_{j \in \J} x'_j$ such that each static flow~$x'_j$ conserves flow on
  all~$v \in V\setminus (\set{\source, \sink_j})$ and $\abs{x'_j} = x'_{e_j} = \frac{\inrate_j}{\inrate}$ for~$j \in
  \J$.
\end{restatable}

\begin{restatable}[cf. Lemma 12 in \cite{sering2018multiterminal}]{lemma}{lemnewarcareactivemultisink} \label{lem:new_arc_are_active:multi_sink}
In a Nash flow over time $\exf$ in $\exG$ all new arcs $(e_j)_j \in J$ are active for all particles $\phi \in \Flow$.
\end{restatable}

Even though the proofs of these lemmas were given in \cite{sering2018multiterminal} we included detailed versions with the new notation in the appendix.

\paragraph{Nash flows over time decomposition.} As the final step (and as an new result), we decompose the
single-commodity Nash flow over time in $\exG$ to obtain a feasible multi-commodity flow over time in the original graph
$G$. For each thin flow phase $I = [\phi_1, \phi_2)$ with thin flow $(x'_e, \l'_e)$ and thin flow decomposition
$(x'_j)_{j \in \J}$ we set
\[f^+_{j,e}(\theta) \coloneqq \frac{x'_{j,e}}{\l'_u} \quad \text{ for } \theta \in [\l_u(\phi_1), \l_u(\phi_2)) \quad 
\text{ and } \quad  f^-_{j,e}(\theta) \coloneqq \frac{x'_{j, e}}{\l'_v} \quad 
\text{ for } \theta \in \l_v(\phi_1), \l_v(\phi_2))\]
for all $j \in \J$ and every $e = uv \in E$. Note that if $\l'_u = 0$ we have $[\l_u(\phi_1), \l_u(\phi_2)) =
\emptyset$. We call the family of functions $(f_{j, e}^+, f_{j, e}^-)$ a \defemph{Nash flow over time decomposition}.

\begin{theorem} \label{thm:nash_flows:multi_sink}
The Nash flow over time decomposition $(f_{j, e}^+, f_{j_e}^-)$ is a multi-commodity Nash flow over time in the
original network with $I_j = [0, \infty)$ for all $j \in \J$.
\end{theorem}

\proof{Proof of \Cref{thm:nash_flows:multi_sink}.}
Throughout this proof $\delta^-_v$ and $\delta^+_v$ denote the incoming and outgoing arcs of $v$ within the original network
$G$ by. Since the particles are partitioned into thin flow phases we consider each thin flow phase $\K = [\phi_1,
\phi_2)$ separately. Let $(x', \l')$ be the corresponding thin flow with thin flow decomposition $(x'_j)_{j \in \J}$.
Furthermore, we denote the interval of local times of particles in $\K$ by $\K_v \coloneqq [\l_v(\phi_1),\l_v(\phi_2))$
for every node $v$.

First, we have to show that the flow over time decompositions form a feasible multi-commodity flow over time. For every
$j \in \J$ and every node $v \in V \setminus \set{\sink_j}$ and all $\theta \in \K_v$ the flow conservation condition
\eqref{eq:flow_conservation:multi_commodity} holds since
\begin{align*}
\sum_{e\in\delta^-_v} f^-_{j, e}(\theta) - \sum_{e\in\delta^+_v} f^+_{j, e}(\theta) &= \sum_{e\in\delta^-_v} 
\frac{x'_{j, e}}{\l'_v} - \sum_{e\in\delta^+_v} \frac{x'_{j,e}}{\l'_v} \\
&= \frac{1}{\l'_v} \cdot \left(
\sum_{e\in\delta^-_v}
x'_{j, e} - \sum_{e\in\delta^+_v} x'_{j,e} \right)
= \begin{cases}
0 & \text{ if } v \in V \setminus \set{\source},\\
\frac{\inrate_j}{\inrate \cdot \l'_{\source}}  = \inrate_j & \text{ if } v = \source.
\end{cases}
\end{align*}
   
Furthermore, for $x'_e > 0$ we obtain for all $\theta \in \K_u$ and $\vartheta = \T_e(\theta) \in \K_v$ that
\[f^-_{j, e}(\theta) = \frac{x'_{j,e}}{\l'_v} = \frac{x'_e}{\l'_v} \cdot
\frac{x'^j_e}{\l'_u} \cdot \frac{\l'_u}{x'_e} = f^-_e(\theta) \cdot \frac{f^+_{j, e}(\vartheta)}{f_e^+(\vartheta)}\]
and for $x'_e = 0$ we have $x'_{j,e} = 0$, which implies $f^-_{j, e}(\theta) = 0$. This shows that
\eqref{eq:FIFO:multi_commodity} is satisfied.
  
\Cref{eq:totaloutflow:multi_commodity} holds since the total flow is a feasible
flow over time, and therefore, we have a feasible multi-commodity flow over time.
  
Finally, since
\[F^+_{j, e}(\l_u(\phi)) = \int_0^{\phi} f_{j, e}^+(\l_u(\xi)) \cdot \l'_u(\xi) \diff \xi 
= \int_0^\phi x'^j_e(\xi) \diff \xi = \int_0^{\phi} f_{j, e}^-(\l_v(\xi)) \cdot \l'_v(\xi) \diff \xi 
= F^-_{j, e}(\l_v(\phi))\]
we obtain by \Cref{lem:nash_flow_characterization:multi_commodity} that $(f_{j, e}^+, f_{j_e}^-)$ is indeed a
multi-commodity Nash flow over time.
\hfill\Halmos\endproof

\section{Conclusion and further research.} \label{sec:conclusion}

Dynamic equilibria, called Nash flows over time, in a single-commodity setting were already understood quite well,
but unfortunately, they are highly unrealistic for most real-world traffic scenarios, where each traffic user has a personal origin and destination. In order to attack this drawback
we considered dynamic equilibria in the multi-commodity version of this problem.
Unfortunately, the essential property of a global FIFO principle, as it is given in the single-commodity setting, does not hold for
these scenarios anymore, and therefore, it is not possible to extend multi-commodity Nash flows over time step by step.
Instead, we have to consider all infinitesimally small players at the same time as the choice of early particles depends
not only on all the flow already in the network but may also depend on all future flow. Even though we could not use any
of the techniques of the single-commodity model, we were still able to show the existence of dynamic equilibria in this
multi-commodity setting with the help of infinite dimensional-variational inequalities. Since this was known before, the
main contribution of this paper is the structural insight into these Nash flows over time, as we could show that their derivatives
have to satisfy a set of conditions similar to the thin flow equations introduced by Koch and Skutella \cite{koch2011nash}. The major difference to the single-commodity
case is that we cannot consider each thin flow isolated anymore, but instead, we have to take into account the flow of
the other commodities (the so-called foreign flow), and therefore, we have to consider all flow from the past and the
future simultaneously. Unfortunately, this still does not give a clear instruction on how to construct Nash flows over
time with multiple commodities algorithmically, however, for the special case that all commodities share the same origin the problem of
constructing a Nash flow over time reduces to the single-commodity case. The same holds true for the other extreme case,
that every particle can start at multiple origins but they all share a common destination.

\new{
\paragraph{Open problems.}
Even though this is a important step for the Nash flow over time model, this is by no means a complete theory yet.
Therefore, we want to give an outlook on further research on this topic.}

\new{First of all, it would be very interesting to better understand the structure of the multi-commodity thin flows. As
mentioned in \Cref{remark:piecewise_constant_thin_flows} we conjecture that they always consists of piece-wise constant
functions $x'_e$ and $\l'_v$ but it would be great to have a proof for this.}

\new{Furthermore, the arrival times of a single-commodity Nash flow over time are in some sense unique (to be precise:
this is only proven for right-continuous functions; see \cite[Theorem 6]{cominetti2015dynamic}). It is an open question
whether this is also the case for multi-commodity Nash flows over times. There are several approaches possible to prove
or disprove this.  For example, there are theorems that states that under some conditions the solution space of a
variational inequality is convex. It is conceivable that this could be helpful to show the uniqueness of the arrival
time functions. In the case that they are not unique it would be sufficient to give a counter-example. So it might be worth to consider more complex examples.} 

\new{An interesting property of every routing game with selfish player is the price of anarchy, which measures the quality
of an worst equilibrium to the optimal flow over time (here in terms of the time the last particle arrives at its destination). For the single-commodity model this is conjectured to be
$\frac{e}{e-1}$ (see \cite{correa2019price}) but it is an completely open question for the multi-commodity setting.}

\new{Clearly, Nash flows over time, and in particular multi-commodity flows over time, remain an interesting but
challenging research field.} 

\newpage

\begin{APPENDIX}{Technical proofs and utility lemmas.}

\technicalpropertiesbase*
\proof{Proof of \Cref{lem:technical_properties:base}.} \phantom{a} \label{proof:technical_properties:base}
\begin{enumerate}
\item[\ref{it:q_equiv_z:base}] This follows directly by the definition of $q_e$.

\item[\ref{it:positive_queue_while_emptying:base}] By \eqref{eq:totaloutflow:multi_commodity} we have that $f_e^-(\xi) \leq \capa_e$
almost everywhere. Since $F_e^+$ is monotonically increasing, we obtain for $0 \leq \xi < q_e(\theta)$ that
\[z_e(\theta + \tau_e\! + \xi) = F_e^+(\theta + \xi) - F_e^-(\theta + \tau_e \! + \xi) 
\geq F_e^+(\theta) - F_e^-(\theta + \tau_e) - \xi \cdot \capa_e > z_e(\theta + \tau_e) - q_e(\theta) \cdot \capa_e = 0.\] 
 
\item[\ref{it:in_equals_out_at_exit_time:base}] Again by \eqref{eq:totaloutflow:multi_commodity} together with
\ref{it:positive_queue_while_emptying:base} we obtain for almost all $\xi \in [\theta + \tau_e, \theta + \tau_e +
q_e(\theta))$ that $f_e^-(\xi) = \capa_e$. Hence,
\[F_e^-(T_e(\theta)) = F_e^-(\theta + \tau_e) + q_e(\theta) \cdot \capa_e = F_e^-(\theta + \tau_e) + z_e(\theta + \tau_e) 
= F_e^+(\theta).\]

\item[\ref{it:equal_exit_times:base}] Intuitively, this holds true since whether a particle enters the queue at time
$\theta_1$ or $\theta_2$ does not influence the exit time, as long as no other flow enters the queue during the interval
$[\theta_1, \theta_2]$ and as long as the queue does not deplete during this time. Formally, this follows since
\[z_e(\xi + \tau_e) = F_e^+(\xi) - F_e^-(\xi + \tau_e) \geq F_e^+(\theta_2) - F_e^-(\theta_2 + \tau_e) 
= z_e(\theta_2 + \tau_e) > 0,\]
and therefore $f_e^-(\xi + \tau_e) = \capa_e$ for almost all $\xi \in [\theta_1, \theta_2]$. Thus,
\begin{align*}
T_e(\theta_1) &= \theta_1 + \tau_e + \frac{F_e^+(\theta_1) - F_e^-(\theta_1 + \tau_e)}{\capa_e} \\
&= \theta_1 + \tau_e +\frac{F_e^+(\theta_2) - F_e^-(\theta_2 + \tau_e) + (\theta_2 - \theta_1) \cdot \capa_e}{\capa_e} 
= T_e(\theta_2).\end{align*}

\item[\ref{it:T_monoton:base}] Consider two points in time $\theta_1 < \theta_2$. By \eqref{eq:totaloutflow:multi_commodity} we have
that $f_e^-(\xi) \leq \capa_e$ almost everywhere, and therefore $F_e^-(\theta_2) - F_e^-(\theta_1) \leq  (\theta_2 -
\theta_1) \cdot \capa_e$. Since $F_e^+$ is monotonically increasing we obtain
\begin{align*}T_e(\theta_1) &= \theta_1 + \tau_e + \frac{F_e^+(\theta_1 - \tau_e) - F_e^-(\theta_1)}{\capa_e} \\
&\leq \theta_1 + \tau_e + \frac{F_e^+(\theta_2 - \tau_e) - F_e^-(\theta_2) + (\theta_2 - \theta_1) \cdot \capa_e}{\capa_e} 
= T_e(\theta_2).\end{align*}

\item[\ref{it:q_is_diffbar:base}] Since $f_e^+$ and $f_e^-$ are locally integrable, Lebesgue's differentiation theorem yields that the integral functions $F_e^+$ and $F_e^-$ are almost everywhere
differentiable. As summation and scaling preserve this property we have that $z_e$, $q_e$ and $T_e$ are almost
everywhere differentiable as well.

\item[\ref{it:derivative_of_q:base}] For almost all $\theta \in \Time$ we have by \eqref{eq:totaloutflow:multi_commodity} that
\[z'_e(\theta + \tau_e) = f_e^+(\theta) - f_e^-(\theta + \tau_e) = \begin{cases}
f_e^+(\theta) - \capa_e & \text{ if } z_e(\theta + \tau_e) > 0,\\
\min\Set{0, f_e^+(\theta) - \capa_e} & \text{ else.}
\end{cases}\]
The claim follows immediately by using \ref{it:q_equiv_z:base}.
\end{enumerate}
\hfill\Halmos\endproof

\nashflowcharacterizationmulticommdoity*
\proof{Proof of \Cref{lem:nash_flow_characterization:multi_commodity}.} \label{proof:nash_flow_chracterization:multi_commodity}
\ref{it:nash_flow:multi_commodity}$\Rightarrow$\ref{it:in_equals_out_at_l:multi_commodity}: Let $\xi \in [0, \phi]$ be
maximal with~$F_{j, e}^+(\l_{j,u}(\xi)) = F_{j,e}^-(\l_{j, v}(\phi))$. Such particle~$\xi$ exists due to the
intermediate value theorem, together with the fact that $F_{j, e}^+ \circ \l_{j, u}$ is continuous
and with the following inequality, which follows by the monotonicity of $F_{j, e}^-$ and
\Cref{lem:fifo:multi_commodity}:
\[F_{j, e}^+(\l_{j, u}(0)) = 0 \quad \leq \quad F_{j, e}^-(\l_{j, v}(\phi))\quad \leq \quad 
F_{j, e}^-(\T_e(\l_{j, u}(\phi)))  = F_{j, e}^+(\l_{j, u}(\phi)).\] 
Note that the second inequality holds because of~$\l_{j, v} (\phi) \leq \T_{j, e}(\l_{j, u}(\phi))$. In the case of $\xi
= \phi$ we are done, so suppose~$\xi < \phi$. For all particles $\varphi \in (\xi, \phi]$ we know that~$\T_e(\l_{j,
u}(\varphi)) \not= \l_{j, v}(\phi)$, since otherwise, we had with \Cref{lem:fifo:multi_commodity} that $F_{j,
e}^+(\l_{j, u}(\varphi)) = F_{j, e}^-(\T_e (\l_{j, u}(\varphi))) = F_{j, e}^-(\l_{j, v}(\phi))$, which would contradict
the maximality of~$\xi$. Hence, $e$ is not active for particles in $(\xi, \phi]$ which implies $f^+_{j, e}(\theta) = 0$
for almost all $\theta \in \l_{j, u}((\xi, \phi]) = (\l_{j, u}(\xi),\l_{j, u}(\phi)]$ by
\eqref{eq:Nash_condition:multi_commodity}. This leads to
\[F_{j, e}^+(\l_{j, u}(\phi)) - F_{j, e}^-(\l_{j, v}(\phi)) = F_{j, e}^+(\l_{j, u}(\phi)) - F_{j, e}^+(\l_{j, u}(\xi)) = 
\int_{\l_{j, u}(\xi)}^{\l_{j, u}(\phi)} f_{j, e}^+(\vartheta) \diff \vartheta = 0,\]
which shows \ref{it:in_equals_out_at_l:multi_commodity}.

\ref{it:in_equals_out_at_l:multi_commodity}$\Rightarrow$\ref{it:nash_flow:multi_commodity}: Consider a particle $\phi$
and an arc $e = uv$ such that $e$ is not active for $\phi$ and $j$, in other words, $\l_{j, v}(\phi) < \T_e(\l_{j,
u}(\phi))$. Then, the continuity of $\l_{j, v}$ and $\T_e \circ \l_{j, u}$ implies that there exists an $\epsilon > 0$
with $\l_{j, v}(\phi + \epsilon) < \T_e(\l_{j, u}(\phi - \epsilon))$ and that $e$ is not active for all particles
in~$[\phi - \epsilon, \phi + \epsilon]$ and $j$. This, the fact that $f_{j, e}^+$ and~$f_{j, e}^-$ are non-negative and
\Cref{lem:fifo:multi_commodity} gives us
\begin{align*} 0 \leq \int_{\l_{j, u}(\phi-\epsilon)}^{\l_{j, u}(\phi+\epsilon)} f_{j, e}^+(\xi) \diff \xi
  &= \int_{\T_e(\l_{j, u}(\phi-\epsilon))}^{\T_e(\l_{j, u}(\phi+\epsilon))} f_{j, e}^-(\xi) \diff \xi\\
  &\leq \int_{\l_{j, v}(\phi + \epsilon)}^{\T_e(\l_{j, u}(\phi+\epsilon))} f_{j, e}^-(\xi) \diff \xi \\
  &= F_{j, e}^-(\T_e(\l_{j, u}(\phi+\epsilon))) - F_{j, e}^-(\l_{j, v}(\phi + \epsilon))\\
  &= F_{j, e}^+(\l_{j, u}(\phi+\epsilon)) - F_{j, e}^-(\l_{j, v}(\phi + \epsilon))\\
  &\hspace{-0.3mm}\stackrel{\text{\ref{it:in_equals_out_at_l:multi_commodity}}}{=} 0.
\end{align*}    
Hence, $f_{j, e}^+(\theta) = 0$ for almost all~$\theta \in [\l_{j, u}(\phi - \varepsilon), \l_{j, u}(\phi +
\varepsilon)]$. In other words, for almost all $\theta \in \Time$ it holds that~$\theta \notin \l_{j,u}(\Phi_{j,e})
\Rightarrow f_{j, e}^+(\theta) = 0$. This is true because for $\theta \geq \l_{j, u}(0)$ we find a particle $\phi$ with
$\l_{j, u}(\phi) = \theta$, due to the fact that $\l_{j, u}$ is surjective. This shows that $f$ is a Nash flow over
time, which finishes the proof.
\hfill\Halmos\endproof

\lemthinflowsdecompositionmultisink*
\proof{Proof of \Cref{lem:thin_flow_decomp:multi_sink}.}
  Let $\Paths$ be the set of all $\source$-$\sink$-paths in the current shortest paths network~$G' = (V, E')$. Note that
  $G'$ is always acyclic and $x'$ can therefore be described by the path vector~$(x'_P)_{P\in \Paths}$ due to the flow
  decomposition theorem. For all
  $j\in \J$ let $\Paths_j$ be the set of all $\source$-$\sink$-paths that contain~$e_j$. These sets form a partition of
  $\Paths$ since every path has to use exactly one of the new arcs. By setting $x'_j \coloneqq \sum_{P \in \Paths_j}
  x'_P\big|_E$ we obtain the desired decomposition of $x'$, because $x'_P\big|_E$ for $P\in \Paths_j$ conserves flow on
  all nodes except for the ones in~$\Set{\source, \sink_j}$ and the same is true for sums of these path flows.
  
  Since $x'_j$ sends $\abs{x'_j}$ flow units from $\source$ over $e_j$ to $\sink_j$ we have $\abs{x'_j} = x'_{e_j}$. It
  remains to show that $x'_{e_j} = \frac{\inrate_j}{\inrate}$ for all~$j \in \J$. Suppose that this is not true. Since
  $x'$ sends exactly $1 = \sum_{j \in \J} \frac{\inrate_j}{\inrate}$ flow units from $\source$ to $\sink$, there has to
  be an index $a \in \J$ with $x'_{e_a} > \frac{\inrate_a}{\inrate}$ and an index $b \in \J$ with~$x'_{e_b} <
  \frac{\inrate_b}{\inrate}$.
  
  With \Cref{lem:bound_for_l':multi_sink} it follows that
  \[\l'_{\sink_b} \leq \frac{1}{\sigma} \stackrel{\eqref{eq:definition_new_arcs:mult_sink}}{<}
  \frac{\inrate_a}{\inrate \cdot \capa_{e_a}} < \frac{x'_{e_a}}{\capa_{e_a}} \stackrel{\eqref{eq:l'_v_tight:multi_source}}{\leq}
  \l'_{\sink} \quad \text{ and } \quad
  \frac{x'_{e_b}}{\capa_{e_b}} \stackrel{\eqref{eq:definition_new_arcs:mult_sink}}{=} 
  \underbrace{\frac{x'_{e_b} \cdot \inrate}{\inrate_b}}_{< 1} \cdot
  \frac{2}{\sigma} < \underbrace{\frac{x'_{e_a}\cdot \inrate}{\inrate_a}}_{> 1} \cdot \frac{2}{\sigma}
  \stackrel{\eqref{eq:definition_new_arcs:mult_sink}}{=} \frac{x'_{e_a}}{ \capa_{e_a}}
  \stackrel{\eqref{eq:l'_v_tight:multi_source}}{\leq} \l'_{\sink}.\]
  
  But this is a contradiction, since \eqref{eq:l'_v_min:multi_source} yields that $\l'_{\sink} = \min\limits_{j \in \J}
  \rho_{e_j}(\l'_{\sink_j}, x'_{e_j})$ and the last two equations show that $\rho_{e_b}(\l'_{\sink_b}, x'_{e_b}) <
  \l'_{\sink}$. Hence, we have $x'_{e_j} = \frac{\inrate_j}{\inrate}$ for all~$j \in \J$, which finishes the proof.
\hfill\Halmos\endproof

\lemnewarcareactivemultisink*
\proof{Proof of \Cref{lem:new_arc_are_active:multi_sink}.}
  For particle $\phi = 0$ there are no queues yet, and therefore, the exit time for each arc $e$ is~$\T_e(\theta) =
  \theta + \tau_e$. Hence, $\l_{\sink_j}(0) = \delta_j$ and by construction we have for all $j\in \J$ that
  \[\l_{\sink}(0) = \l_{\sink_j}(0) + \tau_{e_j} = \T_{e_j}(\l_{\sink_j}(0)).\]
  Therefore, all arcs $e_j$ are active at the beginning and also during the first thin flow phase because by
  \Cref{lem:thin_flow_decomp:multi_sink} we have $x'_{e_j} > 0$ for the first thin flow with resetting, which implies
  that $e_j$ stays active.
  
  Suppose now for contradiction that there are particles for which not all new arcs are active. Let $\phi_0$ be the
  infimum of these particles. By the consideration above we have $\phi_0>0$ and
  \Cref{lem:bound_for_l':multi_sink,lem:thin_flow_decomp:multi_sink} imply that
  \[f^+_{e_j}(\l_{\sink_j}(\phi))=\frac{x'_{e_j}}{\l'_{\sink_j}}
   \geq x'_{e_j} \cdot \sigma = \frac{\inrate_j}{\inrate} \cdot \sigma
   \stackrel{\eqref{eq:definition_new_arcs:mult_sink}}{>} \capa_{e_j}\]
   for almost all $\phi \in [0, \phi_0)$ and all $j\in \J$. Hence,
   \Cref{lem:technical_properties:base}~\ref{it:derivative_of_q:base} together with the fact that $\l'_{t_j} > 0$ (due
   to the positive throughput of $x'$ at $t_j$) yields
  \[\frac{\diff}{\diff \phi}  q_{e_j}(\l_{\sink_j}(\phi)) = q'_{e_j}(\l_{\sink_j}(\phi)) \cdot \l'_{\sink_j}(\phi) > 0.\]
  In other words, a queue is building up within $[0,\phi_0)$, and therefore, $q_{e_j}(\l_{\sink_j}(\phi_0)) > 0$ for all
  $j \in \J$. The continuity of $q_{e_j} \circ \l_{\sink_j}$ implies that there will be positive queues for all $\phi
  \in [\phi_0, \phi_0 + \varepsilon]$ for sufficiently small $\varepsilon > 0$. Since arc with positive queues are
  always active in the single-commodity case, all new arcs are active during this interval contradicting the existence
  of $\phi_0$.
\hfill\Halmos\endproof

\begin{lemma}[Differentiation rule for a minimum]\label{lem:diff_rule_for_min:pre}
For every element $e$ of a finite set $E$ let $T_e \colon \R_{\geq 0} \rightarrow \R$ be a function that is
differentiable almost everywhere and let $\l(\theta) \coloneqq \min_{e \in E} T_e(\theta)$ for all $\theta \geq 0$. It holds
that $l$ is almost everywhere differentiable with
\begin{equation} \label{eq:diff_rule_for_min:pre}
\l'(\theta) = \min_{e \in E'_\theta} T'_e(\theta)
\end{equation}
for almost all $\theta \geq 0$ where $E'_\theta \coloneqq \set{e \in E| \l(\theta) = T_e(\theta)}$.
\end{lemma}

\begin{proof}
Let $\phi \geq 0$ such that all $T_e$, for all $e \in E$, are differentiable, which is almost everywhere.
Since all functions $T_e$ are continuous at $\phi$ we have for sufficiently small $\epsilon > 0$ that $\l(\phi +
\xi) = \min_{e \in E'_\phi} T_e(\phi + \xi)$ for all $\xi \in [\phi, \phi + \epsilon]$. It follows that

\[\lim_{\xi \,\searrow\, 0} \frac{\l(\phi + \xi) - \l(\phi)}{\xi}
= \lim_{\xi \,\searrow\, 0}  \min_{e \in E'_\phi} \frac{T_e(\phi + \xi) - \l(\phi)}{\xi}
= \min_{e \in E'_\phi} \lim_{\xi \,\searrow\, 0} \frac{T_e(\phi + \xi) - T_e(\phi)}{\xi}
= \min_{e \in E'_\phi} T'_e(\phi).\]
Note that every point $\phi$ where all $T_e$ are differentiable, but for which the left derivative of $\l$ does not
coincide with the right derivative of $\l$, is a proper crossing of at least two $T_e$ functions. Therefore, these
points are isolated and form a null set. Hence, we have $\l'(\phi) = \min_{e \in E'_\phi} T'_e(\phi)$ for almost
all~$\phi \in \R_{\geq 0}$.
\end{proof}

\end{APPENDIX}

\section*{Acknowledgments.}
I want to thank my supervisor Martin Skutella, not only for introducing me to this exiting topic of Nash flows over time, but in particular, for giving me advices for the multi-terminal Nash flows over time and for taking the time to listen to all my new ideas and progress.


\bibliographystyle{informs2014} 
\bibliography{literature} 


\end{document}